\RequirePackage{color}
\documentclass[11pt]{article}

\RequirePackage{titling}
\RequirePackage[affil-it]{authblk}
\RequirePackage[letterpaper, left=1.2truein, right=1.2truein, top = 1.2truein,bottom = 1.2truein]{geometry}
\RequirePackage{amsthm,amsmath,amsfonts,amssymb}
\RequirePackage[numbers]{natbib}
\RequirePackage[colorlinks,citecolor=blue,urlcolor=blue,breaklinks=true]{hyperref}
\RequirePackage{booktabs, enumitem}
\RequirePackage{graphicx}
\RequirePackage{tikz}
\usetikzlibrary{positioning,graphs,quotes,shapes,arrows}
\tikzstyle{line}=[draw,-stealth,thick]
\RequirePackage[toc,page]{appendix}

\theoremstyle{plain}
\newtheorem{theorem}{Theorem}
\newtheorem{lemma}{Lemma}
\newtheorem{assumption}{Assumption}

\newcommand{\bbR}{\mathbb R}
\newcommand{\bbP}{\mathbb P}
\newcommand{\calLr}{\mathcal L^r}
\newcommand{\calR}{\mathcal R}
\newcommand{\calH}{\mathcal H}
\newcommand{\calM}{\mathcal M}
\newcommand{\calF}{\mathcal F}
\newcommand{\calJ}{\mathcal J}
\newcommand{\lob}{l^{r}}
\newcommand{\lms}{l^{\overline{r}}}
\newcommand{\Lob}{L^{r}}
\newcommand{\Lms}{L^{\overline{r}}}
\newcommand{\Lobi}{L^r_i}
\newcommand{\odds}{\mathrm{Odds}}
\newcommand{\oddsr}{\mathrm{Odds}^{r}}
\newcommand{\Phir}{\Phi^{r}}
\newcommand{\phir}{\phi^{r}}
\newcommand{\bfDr}{\mathbf D^{r}}
\newcommand{\alphar}{\alpha^{r}}
\newcommand{\PEN}{\mathtt{PEN}}
\newcommand{\tr}{^\intercal}
\newcommand{\quads}{\quad\ }
\newcommand{\E}{\mathbb{E}}
\newcommand{\domr}{\mathrm{dom}^r}
\newcommand{\bbG}{\mathbb G}
\newcommand{\calT}{\mathcal T}
\newcommand{\calG}{\mathcal G}
\newcommand{\calO}{\mathcal O}
\newcommand{\calU}{\mathcal U}
\newcommand{\calN}{\mathcal N}
\newcommand{\bfA}{\mathbf A}
\newcommand{\bfB}{\mathbf B}
\newcommand{\bfC}{\mathbf C}
\newcommand{\bfD}{\mathbf D}
\newcommand{\bfE}{\mathbf E}
\newcommand{\bfF}{\mathbf F}

\title{Balancing Weights for Non-monotone Missing Data}
\author[1]{Jianing Dong}
\author[2]{Raymond K. W. Wong}
\author[3]{Kwun Chuen Gary Chan}
\affil[1,2]{Department of Statistics, Texas A\&M University}
\affil[3]{Department of Biostatistics, University of Washington}
\date{}

\begin{document}
	\maketitle
	
	\begin{abstract}
		Balancing weights have been widely applied to single or monotone missingness due to empirical advantages over likelihood-based methods and inverse probability weighting approaches. This paper considers non-monotone missing data under the complete-case missing variable condition (CCMV), a case of missing not at random (MNAR). Using relationships between each missing pattern and the complete-case subsample, we construct a weighted estimator for estimation, where the weight is a sum of ratios of the conditional probability of observing a particular missing pattern versus that of observing the complete-case, given the variables observed in the corresponding missing pattern. However, plug-in estimators of the propensity odds can be unbounded and lead to unstable estimation. Using further relations between propensity odds and balancing of moments across response patterns, we employ tailored loss functions, each encouraging empirical balance across patterns to estimate propensity odds flexibly using a functional basis expansion. We propose two penalizations to control propensity odds model smoothness and empirical imbalance. We study the asymptotic properties of the proposed estimators and show that they are consistent under mild smoothness assumptions. Asymptotic normality and efficiency are developed. Simulation results show the superior performance of the proposed method.
	\end{abstract}
	

	\section{Introduction}
	\label{s:intro}
	Missing data is ubiquitous in many research fields, such as health sciences, economics, and sociology. Monotone missing patterns, which arise in longitudinal studies and multi-stage sampling, have been widely studied. In practice, non-monotone missing patterns are more common since an ordering of variables may not exist to characterize the response patterns. For example, participants of an observational cohort may miss a particular study visit but return in later visits. 
	
	Restricting the analysis to complete cases in which all relevant variables are observed is a convenient solution but ignores information from partially missing data. It may lead to biased estimates unless data are missing completely at random (MCAR) \citep{little2019statistical}, and would typically lose efficiency even under MCAR. The missing at random (MAR) assumption, where missingness depends only on observed data, is more reasonable and is widely used to handle missing data in likelihood-based or weighting methods.
	
	Researchers often focus on a specific target, such as the coefficients of linear regression with Gaussian errors, which can be easily estimated by maximum likelihood estimation if the full-data density is known. To estimate parameters in the full-data density where the full variable vector $L$ may be partially missing, \citet{little1985maximum} suggested a mixture of normal distributions as the model of missing data density. However, its validity heavily relies on parametric density assumptions. The likelihood-based methods are attractive because the missingness probability can be factored out from the likelihood under MAR \citep{rubin1976inference}. To mimic the score function (the derivative of log-likelihood) in the absence of missing data, \citet{reilly1995mean} proposed a mean-score approach, and \citet{chatterjee2003pseudoscore} proposed a pseudo-score approach for the bivariate case and extended to partial questionnaire design which is a particular case of non-monotone missing data \citep{chatterjee2010inference}. \citet{chen2004nonparametric} proposed a semi-parametric method based on reparametrization of the joint likelihood into a product of the odds-ratio functions. However, these methods cannot directly handle continuous variables and require the discretization of continuous variables.
	
	Inverse propensity weighting (IPW) methods \citep{robins1994estimation} are also attractive since they do not require the specification of full-data likelihood. In contrast, a propensity model for non-missingness is needed. However, the resulting estimators can be highly unstable since even a small set of tiny propensity estimates would lead to extreme weights. A more direct weighting approach, the balancing method, is developed especially for the average treatment effect (ATE), which can be viewed as a missing data problem with only the response variable missing; see, for example, \citet{zubizarreta2015stable}, \citet{chan2016globally} and \citet{wong2018kernel}, and \citet{zhao2019covariate} introduced tailored loss functions to show the connection between IPW and balancing methods. These balancing methods handle the single missing variable scenario, which has yet to be extended to handle non-monotone missingness. 
	
	\citet{robins1997non} argued that missing not at random (MNAR) is more natural than MAR under non-monotone missingness. However, many different specifications of MNAR are not identifiable from observed data. \citet{little1993pattern} imposed an identifying condition, the complete-case missing variable (CCMV) restriction, that matches the unidentified conditional distribution of missing variables for missing patterns to the identified distribution for complete cases. Then, the full-data density can be estimated by a pattern-mixture model. IPW estimators were also proposed under CCMV restriction and extensions \citep{tchetgen2018discrete}. Another increasingly popular MNAR mechanism is the no self-censoring (NSC) restriction, studied in \citet{sinha2014semiparametric}, \citet{sadinle2017itemwise} and \citet{malinsky2022semiparametric}. Under MNAR, specifications of missing variable density typically rely on parametric assumptions and have intrinsic drawbacks. For instance, likelihood-based methods are often restricted to discrete variables due to complex integral computation, and IPW methods are often unstable due to difficulties in weight estimation for multiple missing patterns. 
	
	We propose a non-parametric approach that generalizes the balancing approaches to non-monotone missing scenarios under the CCMV assumption. Based on a tailored loss function, the proposed method approximately balances a set of functions of observed variables and mimics the missing-data adjustment of the mean-score approaches. The proposed method does not require the discretization of continuous variables and is less prone to misspecification due to the non-parametric modeling. A carefully designed penalization strikes a good balance between bias and variance, and prioritizes balancing smoother functions. We also show the consistency and asymptotic efficiency of the resulting estimator.

	\section{Notation and preliminaries}
	\label{s:notation}
	We formally describe the setup of the problem. Let $L=(L^{[1]},\ldots,L^{[d]})\in\bbR^d$ be a random vector of interest and $R=(R^1,\ldots,R^d)\in\calR\subseteq\{0,1\}^d$ be a binary random vector where $R^j=1$ indicates that $L^{[j]}$ is observed and $\calR=\{r:P(R=r)>0\}$ is the set of possible response patterns in the study. Denote the complete-case pattern by $1_d=(1,\ldots,1)$. Let $M=|\calR|$ be the number of response patterns. For the response pattern $r$, we denote the observed variables by $\Lob\in\mathbb{R}^{d_r}$ and the missing variables by $\Lms\in\mathbb{R}^{d-d_r}$ where $d_r$ is the number of observed variables. So, the observations are $(L^{R_i}_i,R_i)_{i=1}^N$. An example of a bivariate non-monotone missing structure with $M=3$ is 
	\begin{center}
		\renewcommand{\arraystretch}{1.5}
		\begin{tabular}{cccc}
			\toprule
			response pattern & $R=11$ & $R=10$ & $R=01$ \\
			\hline
			observed variables & $L^{11}=(L^{[1]},L^{[2]})$ & $L^{10}=L^{[1]}$ & $L^{01}=L^{[2]}$ \\
			\hline
			missing variables & $L^{\overline{11}}=\emptyset$ & $L^{\overline{10}}=L^{[2]}$ & $L^{\overline{01}}=L^{[1]}$ \\
			\bottomrule
		\end{tabular}\mbox{} 
	\end{center}
	Let $\theta_0\in\bbR^q$ be the parameter of interest which is the unique solution to $\E\{\psi_\theta(L)\}=0$, with a known vector-valued estimating function $\psi_\theta(L)=\psi(L,\theta)$ with values in $\bbR^q$. For instance, we could use the quasi-likelihood estimating functions for the generalized linear models. If full data were observed, a solution to the estimating equations $N^{-1}\sum_{i=1}^N\psi_\theta(L_i)=0$ is a common Z-estimator. However, $\psi_\theta(L_i)$ can only be evaluated at complete samples, and a complete-case analysis is typically biased or inefficient. To address the problem, we posit the following assumptions throughout this paper.
	
	\begin{assumption}\label{assump1}$ $ 
		\begin{enumerate}[label={\textbf{~\Alph*:}},ref={Assumption~\theassumption.\Alph*},leftmargin=1cm]
			\item\label{assump-1A}
			The estimating function $\psi(L,\theta)$ is differentiable with respect to $\theta$ with derivative $\dot{\psi}_\theta(L)$. Also, $\E\{\psi_\theta(L)\}$ has the unique root $\theta_0$ and is differentiable at $\theta_0$ with nonsingular derivative $D_{\theta_0}$.
			\item\label{assump-1B}
			There exists a constant $\delta_0>0$ such that $P(R=1_d\mid\lob)\ge\delta_0$ for any $r\in\calR$ and so $1_d\in\calR$.
			\item\label{assump-1C}
			$p(\lms\mid\lob,R=r)=p(\lms\mid\lob,R=1_d)
			$, or equivalently, $\frac{P(R=r\mid l)}{P(R=1_d\mid l)}=\frac{P(R=r\mid\lob)}{P(R=1_d\mid\lob)}$, for any $r\in\calR\backslash 1_d$\ .
		\end{enumerate}
	\end{assumption}
	\ref{assump-1A} is a standard regularity assumption for Z-estimation. \ref{assump-1B} ensures that complete cases are available for analysis. The first equation in \ref{assump-1C} is the CCMV condition, while the second equation motivates the logit discrete choice nonresponse model (LDCM) in \citet{tchetgen2018discrete}. To see that the missing mechanism is MNAR, we define the propensity as the probability of data belonging to a specific response pattern conditional on the variables of interest and write it as $P(R=r\mid l)$.
	The propensity odds between patterns $r$ and $1_d$
	is defined as $P(R=r\mid l)/P(R=1_d\mid l)$. The CCMV condition states that the propensity odds
	depends on $l$ via $\lob$, and so is equal to
	$\oddsr(\lob)=P(R=r\mid\lob)/P(R=1_d\mid\lob)$. The propensity is a function of $L=\bigcup_{s\in\calR\backslash 1_d}L^s$ and does not satisfy either MCAR or MAR conditions because propensities and odds are related through
	\begin{align}\label{propensity}
		P(R=r\mid l)
		=\frac{P(R=r\mid l)/P(R=1_d\mid l)}{\sum_{s\in\calR}P(R=s\mid l)/P(R=1_d\mid l)}
		=\frac{\oddsr(\lob)}{\sum_{s\in\calR}\odds^s(l^s)}\ .
	\end{align}
	\citet{tchetgen2018discrete} showed that Assumption \ref{assump1} is sufficient for non-parametric identification and developed the inverse propensity weighting (IPW) estimator that is motivated by the law of total expectation:
	\begin{align}\label{total-expectation}
		\E\left\{\frac{\mathsf{1}_{R=1_d}\psi_\theta(L)}{P(R=1_d\mid L)}\right\}
		=\E\left[\E\left\{\frac{\mathsf{1}_{R=1_d}\psi_\theta(L)}{P(R=1_d\mid L)}\bigg\vert L\right\}\right]
		=\E\{\psi_\theta(L)\}\ .
	\end{align} 
	
	Taking the reciprocal of \eqref{propensity} and that $\odds^{1_d}(l)=1$, 
	\begin{align}\label{sumweights}
		\frac{1}{P(R=1_d\mid l)}=\sum_{r\in\calR}\oddsr(\lob)\ .
	\end{align}
	As such, one can focus on the estimation of the odds. A standard approach is to fit a logistic regression where
	the binary response indicates whether the pattern is
	$R=1_d$ or $R=r$, and obtain an estimate of $P(R=r\mid\lob,R\in\{1_d,r\})$ \citep{tchetgen2018discrete}. Due to the relationship
	\begin{align*}
		\oddsr(\lob)=\frac{P(R=r\mid\lob,R\in\{1_d,r\})}{1-P(R=r\mid\lob,R\in\{1_d,r\})}\ ,
	\end{align*}
	a plug-in estimator of $\oddsr(\lob)$ can then be constructed based on the predicted probability.
	
	Based on \eqref{total-expectation} and \eqref{sumweights}, the resulting estimator $\hat{\theta}$ can be obtained by solving the weighted estimating equations $N^{-1}\sum_{i=1}^N\mathsf{1}_{R_i=1_d}\sum_{r\in\calR}\hat{w}^r(\Lobi)\psi_\theta(L_i)=0$ where $\hat{w}^r(\lob)$ represents an estimator of $\oddsr(\lob)$ and $\mathsf{1}_A$ is an indicator function of an event $A$. However, the estimated odds may lead to extremely large weights and hence an unstable estimator of $\theta_0$ when the likelihood estimation of a misspecified missingness model is used. Similar phenomena were observed in inverse propensity estimation in \cite{kang2007demystifying} and many subsequent papers, which motivates covariate balancing methods for parameter estimation. See, for example, \cite{imai2014covariate}, \cite{yiu2018covariate}, \cite{zhao2019covariate} and \cite{tan2020regularized}.
	
	We note that the estimation of $P(R=1_d\mid l)$ or $\oddsr(\lob)$ is not the ultimate target but only serves as an intermediate step for estimating $\theta_0$. Instead of using the entropy loss, \emph{i.e.}, the negative log-likelihood function, as the loss function, which is typical for logistic regression, we employ a tailored loss that directly imposes empirical control on key identifying conditions of the weights, which we call the balancing conditions described below. These balancing conditions are directly related to estimating target parameters $\theta_0$. Since the missingness mechanism is typically not our interest, parsimonious modeling may not be required. Instead, we will use functional basis expansions for flexible modeling, which requires penalizations for stable estimation. In addition, our proposed estimator achieves the semiparametric efficiency bounds of parameters defined through the estimating equations. By contrast, \citet{tchetgen2018discrete} only provided asymptotic variances when the parametric model of either propensity odds or missing variable density is correctly specified. Here, we present the semiparametric efficiency bound,
	whose proof is given in Appendix \ref{sec:proof-efficiency}. The ``regular estimators'' are defined according to \citet{begun1983information}. The regularity of an estimator can be viewed as a robustness or stability property.
	\begin{theorem}\label{efficiency-bound}
		Under Assumption \ref{assump1}, the asymptotic variance lower bound for all regular estimators of $\theta_0$ is $D_{\theta_0}^{-1}V_{\theta_0}D_{\theta_0}^{-1\tr}$, where $V_\theta=\E\{F_\theta(L,R)F_\theta(L,R)\tr\}$ and
		{\small
			\begin{align*}
				F_\theta(L,R)=\mathsf{1}_{R=1_d}\sum_{r\in\calR}\oddsr(\Lob)\{\psi_\theta(L)-u_\theta^r(\Lob)\}+\sum_{r\in\calR}\mathsf{1}_{R=r}u_\theta^r(\Lob)-\E\{\psi_\theta(L)\}
			\end{align*}
		}
		with $u_\theta^r(\lob)=\E\{\psi_\theta(L)\mid\Lob=\lob,R=r\}$, the conditional expectation of the estimating function given variables $\Lob$ and is equal to $\E\{\psi_\theta(L)\mid\Lob=\lob,R=1_d\}$ under \ref{assump-1C}. In a slight abuse of notation, $u_\theta^{1_d}(l^{1_d})=\psi_\theta(l)$.
	\end{theorem}

	\section{Construction of our method}
	\label{s:construction}
	In this section, we define the balancing conditions and propose a tailored loss function whose expectation is minimized when the propensity odds model satisfies the balancing conditions. We will show that the imbalance can be explicitly controlled under penalization. Eventually, we will introduce two penalizations to control empirical imbalance and the smoothness of the propensity odds estimates, respectively. 
	
	\subsection{Balancing conditions and benefits}
	To estimate the propensity odds in a way more related to the estimation of $\theta_0$, we first consider the effect of a set of generic weights $w^r$ in the weighted estimating equations that are constructed to estimate $\E\{\psi_\theta(L)\}$. The law of total expectation \eqref{total-expectation} shows that the weighted average in the fully observed pattern is equal to the unweighted average in the population. Moreover, for each missing pattern $r$, one can see that under the CCMV assumption, the following equation holds for any measurable function $g(\lob)$. We call it the balancing condition associated with function $g(\lob)$ using $\oddsr(\lob)$: 
	\begin{align}\label{balancing-condition-expectation}
		\E\left\{\mathsf{1}_{R=r}g(\Lob)\right\}
		&=\E\left\{\mathsf{1}_{R=1_d}\oddsr(\Lob)g(\Lob)\right\}\\
		&=\E\left\{P(R=1_d\mid\Lob)\frac{P(R=r\mid\Lob)}{P(R=1_d\mid\Lob)}g(\Lob)\right\}\nonumber\ .
	\end{align}
	Note that the weights $w^r(\Lob)$ are equal to $\oddsr(\Lob)$ almost surely with respect to the conditional measure given $R=1_d$ if they satisfy the equations $\E\{\mathsf{1}_{R=r}g(\Lob)\}=\E\{\mathsf{1}_{R=1_d}w^rg(\Lob)\}$ for all measurable functions. In other words, the set of balancing conditions associated with all measurable functions identifies the propensity odds. This motivates another way to estimate the propensity odds by balancing the shared variables between two patterns. To explicitly see how balancing approach helps the estimation of $\theta$, we study the error $N^{-1}\sum_{i=1}^N\mathsf{1}_{R_i=1_d}\sum_{r\in\calR}w^r_i\psi_\theta(L_i)-\E\{\psi_\theta(L)\}$ based on a generic set of weights $w^r_i, i=1,\ldots, N$. The error can be first decomposed as
	\begin{align*}
		&\quads\frac{1}{N}\sum_{i=1}^N\mathsf{1}_{R_i=1_d}
		\sum_{r\in\calR}w^r_i\psi_\theta(L_i)-\E\{\psi_\theta(L)\}\\
		&=\sum_{r\in\calR}\left[\frac{1}{N}\sum_{i=1}^N\mathsf{1}_{R_i=1_d}w^r_i\psi_\theta(L_i)-\E\{\mathsf{1}_{R=r}\psi_\theta(L)\}\right]\ .
	\end{align*}
	Further decomposition of each inner term leads to:
	\begin{align}
		&\quads\frac{1}{N}\sum_{i=1}^N\mathsf{1}_{R_i=1_d}w^r_i\psi_\theta(L_i)-\E\{\mathsf{1}_{R=r}\psi_\theta(L)\}\nonumber\\
		&=\frac{1}{N}\sum_{i=1}^N\mathsf{1}_{R_i=1_d}w^r_i\left\{\psi_\theta(L_i)-u_\theta^r(\Lobi)\right\}\label{decomposition1}\\
		&\quads+\frac{1}{N}\sum_{i=1}^N\left\{\mathsf{1}_{R_i=1_d}w^r_iu_\theta^r(\Lobi)-\mathsf{1}_{R_i=r}u_\theta^r(\Lobi)\right\}\label{decomposition2}\\
		&\quads+\frac{1}{N}\sum_{i=1}^N\mathsf{1}_{R_i=r}u_\theta^r(\Lobi)
		-\E\{\mathsf{1}_{R=r}u_\theta^r(\Lob)\}\ .\label{decomposition3}
	\end{align}
	To control the error, we aim to design weights that can control the magnitude of the components in the above decomposition. First, we note that only the first two terms \eqref{decomposition1} and \eqref{decomposition2} depend on the weights. Indeed, the last term \eqref{decomposition3} is expected to converge to zero uniformly over $\theta$ in a compact set $\Theta$ under mild assumptions. Next, we focus on the term \eqref{decomposition2}. It is the empirical version of imbalance associated with $u_\theta^r$. Therefore, to control this error component, we would like the weights to achieve empirical balance, at least approximately. It is the fundamental motivation of the proposed weights. Finally, the term \eqref{decomposition1} with the proposed weights introduced later can be shown to uniformly converge to zero by some technical arguments. In Section \ref{s:asymptotic} and Appendices \ref{sec:proof-efficiency}--\ref{sec:lemma}, we provide rigorous theoretical statements and corresponding proofs. In the next subsection, we propose a tailored loss function that encourages the empirical balance.
	
	\subsection{Tailored loss function for balancing conditions}
	For each missing pattern $r\neq 1_d$, suppose
	\begin{align*}
		\oddsr(\lob;\alphar)=\exp\left\{\Phir(\lob)\tr\alphar\right\}\ ,
	\end{align*}
	where $\Phir(\lob)=\left\{\phi^r_1(\lob),\ldots,\phi^r_{K_r}(\lob)\right\}$ are $K_r$ basis functions for the observed variables in pattern $r$. The basis functions for $\lob$ can be constructed as tensor products of the basis functions for each variable. One may choose suitable basis functions depending on the observed variables and the number of observations in different patterns. For each missing pattern $r\neq 1_d$, we propose the tailored loss function:
	\begin{align*}
		\calLr\{\oddsr(\lob;\alphar),R\}=\mathsf{1}_{R=1_d}\oddsr(\lob;\alphar)-\mathsf{1}_{R=r}\log\oddsr(\lob;\alphar)\ .
	\end{align*}
	Assuming the exchangeability of taking derivatives with respect to $\alphar$ and taking expectations with respect to $L$, the minimizer of $\E[\calLr\{\oddsr(\Lob;\alphar),R\}]$ satisfies $$\partial\E[\calLr\{\oddsr(\Lob;\alphar),R\}]/\partial\alphar=0,$$ which can be rewritten as the balancing conditions: 
	\begin{align*}
		\E\left\{\mathsf{1}_{R=1_d}\oddsr(\Lob;\alphar)\Phir(\Lob)\right\}
		=\E\left\{\mathsf{1}_{R=r}\Phir(\Lob)\right\}\ .
	\end{align*}
	In practice, the minimum tailored loss estimator of $\alphar$, denoted by $\breve{\alpha}^r$, is obtained by minimizing the average loss 
	\begin{align*}
		\calLr_N(\alphar)=N^{-1}\sum_{i=1}^N\calLr\{\oddsr(\Lobi;\alphar),R_i\}\ .
	\end{align*}
	The estimating equations $\nabla\calLr_N(\breve{\alpha}^r)=0$ can be rewritten as the empirical version of balancing conditions using $\oddsr(\Lobi;\breve{\alpha}^r)$:
	\begin{align}\label{balancing-condition-odds}
		\frac{1}{N}\sum_{i=1}^N\mathsf{1}_{R_i=1_d}\oddsr(\Lobi;\breve{\alpha}^r)\Phir(\Lobi)
		=\frac{1}{N}\sum_{i=1}^N\mathsf{1}_{R_i=r}\Phir(\Lobi)\ .
	\end{align}
	Recall that we want to control the term \eqref{decomposition2}. If $u_\theta^r\in\mathrm{span}\{\phi^r_1,\ldots,\phi^r_{K_r}\}$, the empirical balances associated with the basis functions imply the empirical balance associated with $u_\theta^r$. Therefore, the basis functions should be chosen to approximate $u_\theta^r$ well. Theoretically, one can increase the number of basis functions with sample size to extend the spanned space and lower the approximation error. 
	
	\subsection{Penalized optimization and empirical imbalance}
	As mentioned in the last subsection, one wants to balance as many basis functions as necessary to enlarge the space spanned by the basis functions since $u_\theta^r$ is unknown. However, the unpenalized optimizations can result in overfitting or even being unfeasible to compute if one chooses too many basis functions. More precisely, $\calLr_N(C\alphar)\to-\infty$ as $C\to\infty$ if there exists $\alphar$ such that the linear component ${\Phir}(\lob)\tr\alphar>0$ for all data in pattern $r$ and ${\Phir}(\lob)\tr\alphar<0$ for all data in pattern $1_d$.
	
	Therefore, we consider the penalized average loss $\calLr_\lambda(\alphar)=\calLr_N(\alphar)+\lambda J^r(\alphar)$ where the penalty function $J^r(\cdot)$ is a continuous non-negative convex function, and the tuning parameter $\lambda\ge0$ controls the degree of penalization chosen by a cross-validation procedure to be discussed later. Since the estimation of propensity odds is not the final goal, we want to study how penalization affects the estimation of $\E\{\psi_\theta(L)\}$. We will explore the empirical imbalance of $u_\theta^r$, which depends on those of basis functions.
	
	Note that the tailored loss function is convex in $\alphar$. One can consider the sub-differential, as the penalty $J^r(\cdot)$ may be non-differentiable. Denote the sub-differential of $\calLr_\lambda$ at $\alphar$ as $\partial\calLr_\lambda(\alphar)$, which is a set of sub-derivatives $g\in\bbR^{K_r}$ satisfying 
	\begin{align*}
		\calLr_\lambda(\beta^r)\ge\calLr_\lambda(\alphar)+g\tr(\beta^r-\alphar), \textrm{ for all }\beta^r\in\bbR^{K_r}\ .
	\end{align*}
	The sub-differential of $\calLr_\lambda$ at any bounded minimizer $\hat{\alpha}^r$ should include the vector $\mathbf{0}$ since $\calLr_\lambda(\alphar)\ge \calLr_\lambda(\hat{\alpha}^r)+\mathbf{0}\tr(\alphar-\hat{\alpha}^r)$ for all $\alphar\in\bbR^{K_r}$. Therefore, there exists a sub-derivative $\hat{s}=(\hat{s}_1,\ldots, \hat{s}_{K_r})\in\bbR^{K_r}$ of $J^r$ at $\hat{\alpha}^r$ such that 
	\begin{equation}\label{eq:subdiff}\nabla\calLr_N(\hat{\alpha}^r)+\lambda\hat{s}=\mathbf{0} \ .
	\end{equation}

	\subsubsection{Controlling empirical imbalance}
	We denote the absolute difference between the left-hand side and the right-hand side of \eqref{balancing-condition-odds} the empirical imbalance. One immediate observation from \eqref{eq:subdiff} is that the empirical balance \textit{might} not hold exactly when $\lambda\neq 0$. Since the penalty determines $\hat{\alpha}^r$ and $\hat{s}$, one wants to choose an appropriate penalty to ensure control of the empirical imbalance. As a continuous function leads to bounded sub-derivatives, a simple choice is $\ell_1$-norm penalty, \emph{i.e.}, $\sum_{k=1}^{K_r}|\alphar_k|$, whose entries of sub-derivatives belong to $[-1,1]$, and so, for each $k\in\{1,\ldots,K_r\}$,
	\begin{align*}
		\left|\frac{1}{N}\sum_{i=1}^N\mathsf{1}_{R_i=r}\phir_k(\Lobi)
		-\frac{1}{N}\sum_{i=1}^N\mathsf{1}_{R_i=1_d}\oddsr(\Lobi;\hat{\alpha}^r)\phir_k(\Lobi)\right|
		\le\lambda\ .
	\end{align*}
	As such, $\lambda$ uniformly controls all empirical imbalances of basis functions. However, to emphasize the control of empirical imbalance of $u_\theta^r$, the tolerance to the empirical imbalances of different basis functions should be allowed to vary. We should tolerate smaller empirical imbalances for those basis functions that are believed to have stronger approximating power for $u_\theta^r$. For example, if one believes that $u_\theta^r$ satisfies certain smoothness assumptions, one may want to prioritize the empirical balances of smoother basis functions. One natural way to achieve this is to use the weighted $\ell_1$-norm penalty, \emph{i.e.}, $\sum_{k=1}^{K_r}t_k|\alphar_k|$ where $t_k\ge0,k=1,\ldots,K_r$, represent relative tolerance which should be determined by the assumptions of $u_\theta^r$. For each $k\in\{1,\ldots,K_r\}$, the corresponding imbalance bound becomes:
	\begin{align*}
		\left|\frac{1}{N}\sum_{i=1}^N\mathsf{1}_{R_i=r}\phir_k(\Lobi)
		-\frac{1}{N}\sum_{i=1}^N\mathsf{1}_{R_i=1_d}\oddsr(\Lobi;\hat{\alpha}^r)\phir_k(\Lobi)\right|
		\le\lambda t_k\ .
	\end{align*}
	Small $t_k$ should be assigned to important basis functions to ensure stronger balance. A detailed example of the choice of $t_k$ is presented in Section \ref{s:tuning}.
	
	\subsubsection{Controlling propensity odds smoothness}
	In addition to the empirical balance, we could also introduce a penalty to promote the smoothness of the propensity odds estimate. In this paper, we focus on a quadratic penalty of the form ${\alphar}\tr\bfDr\alphar$, where $\bfDr$ is a positive semi-definite matrix.
	As suggested by \citet{wood2017generalized}, the most convenient penalties to control the degree of smoothness are those that measure function roughness as a quadratic form in the coefficients of the function.
	A detailed example of the choice is presented in Section \ref{s:tuning}, where the roughness is related to the order of derivatives of a function. It is worth noting that our asymptotic theory only depends on the positive semi-definite property of $\bfDr$, but not a specific choice of $\bfDr$.

	\subsubsection{A combined penalty}
	\label{s:tuning}
	We suggest the following penalized loss to combine the strengths of the above two penalties:
	\begin{align}\label{penalized-loss}
		\calLr_\lambda(\alphar)
		=\frac{1}{N}\sum_{i=1}^N\calLr\{\oddsr(\Lobi;\alphar),R_i\}
		+\lambda\left\{\gamma \sum_{k=1}^{K_r}t_k|\alphar_k|+(1-\gamma){\alphar}\tr\bfDr\alphar
		\right\},
	\end{align}
	where $0\le\gamma\le 1$ is the tuning parameter controlling the weights between two penalties. A method for tuning parameter selection is given in Section \ref{sec:cross-validation}.
	
	Here, we focus on a standard choice of $\{t_k\}_{k=1}^{K_ r}$ and matrix $\bfDr$,
	based on a prior belief that $u_\theta^r$ and $\log\oddsr$ are smooth. There are many ways to define the roughness of a function. Here, we adopt a version in \citet{wood2017generalized} for twice continuously differentiable functions of $d$-dimensional continuous variables, such as cubic splines. The roughness of function $f$ is defined as $\PEN_2(f)=\langle f,f\rangle$ where the semi-inner product is defined as:
	{\footnotesize
		\begin{align*}
			\langle f,g \rangle
			=\int\cdots\int_{\bbR^d}\sum_{v_1+\cdots+v_d=2}\frac{2!}{v_1!\cdots v_d!}
			\left(\frac{\partial^2f}{\partial x_1^{v_1}\cdots\partial x_d^{v_d}}\right)
			\left(\frac{\partial^2g}{\partial x_1^{v_1}\cdots\partial x_d^{v_d}}\right)
			dx_1\cdots dx_d\ .
		\end{align*} 
	}
	Let $\bfDr$ be the Gram matrix where $\bfDr_{i,j}=\langle\phir_i,\phir_j\rangle$. Then, the roughness of basis expansion ${\Phir}\tr\alphar$ is the quadratic term ${\alphar}\tr\bfDr\alphar$. We orthogonalize the basis functions with respect to the above inner product such that the matrix $\bfDr$ becomes diagonal. Thus, the quadratic term becomes $\sum_{k=1}^{K_r}({\alphar_k})^2\PEN_2(\phir_k)$ and has a simple derivative and is strongly convex with respect to $\alphar_k$ for $k$ such that $\PEN_2(\phir_k)>0$. To prioritize balancing smoother basis functions, we tolerate larger imbalances for rougher functions. Thus, we can connect the tolerance with function roughness. In particular, we choose $t_k=\sqrt{\langle \phir_k,\phir_k\rangle}=\sqrt{\PEN_2(\phir_k)}$. Then, for the minimizer $\hat{\alpha}^r$ of $\calLr_\lambda(\alphar)$, the empirical imbalance of any linear combination of basis functions $v(\cdot):= \sum_{k=1}^{K_r} \beta_k^r \phi_k^r(\cdot)$ is
	\begin{align*}
		&\quad\left|\sum_{k=1}^{K_r}\beta^r_k\left[\frac{1}{N}\sum_{i=1}^N
		\left\{\mathsf{1}_{R_i=r}-\mathsf{1}_{R_i=1_d}\oddsr(\Lobi;\hat{\alpha}^r)\right\}\phir_k(L_i^r)\right]\right|\\
		&\le\lambda\gamma\sum_{k=1}^{K_r}|\beta^r_k|\sqrt{\PEN_2(\phir_k)}
		+2\lambda(1-\gamma)\sum_{k=1}^{K_r}
		|\beta^r_k\hat{\alpha}^r_k|
		\PEN_2(\phir_k)\\
		&\le\lambda\gamma\sqrt{K_r}\sqrt{\sum_{k=1}^{K_r}\PEN_2(\beta^r_k\phir_k)}
		+2\lambda(1-\gamma)\sqrt{\sum_{k=1}^{K_r}\PEN_2(\beta^r_k\phir_k)}
		\sqrt{\sum_{k=1}^{K_r}\PEN_2(\hat{\alpha}^r_k\phir_k)}\\
		&=\lambda\left\{\gamma\sqrt{K_r}
		+2\lambda(1-\gamma)\sqrt{\PEN_2\left(\sum_{k=1}^{K_r}\hat{\alpha}^r_k\phir_k\right)}\right\}\sqrt{\PEN_2\left(\sum_{k=1}^{K_r}\beta^r_k\phir_k\right)}\\
		&=\lambda\left\{\gamma\sqrt{K_r}+2(1-\gamma)\sqrt{\PEN_2({\Phir}\tr\hat{\alpha}^r)}\right\}\sqrt{\PEN_2(v)}
	\end{align*} 
	where $\Phir(\lob)\tr\hat{\alpha}^r=\log\oddsr(\lob;{\hat{\alpha}^r})$ denotes the log transformation of the propensity odds model, and $K_r$ is the number of basis functions. As such, the empirical imbalance of $v$ is proportionally bounded by the square root of its roughness $\PEN_2(v)$. So, smoother functions in the spanned space have smaller empirical imbalances. We assume that $u_\theta^r$ is well approximated by the basis functions under certain smoothness assumptions introduced in \ref{assump-3A}. Thus, the approximation errors are controlled, and the empirical imbalance of $u_\theta^r$ is also controlled.
	
	Besides, $\PEN_2({\Phir}\tr\hat{\alpha}^r)$ is controlled through the quadratic term in the penalty. The rough propensity odds models are penalized more. In total, we control the smoothness of the propensity odds model and the imbalance of $u_\theta^r$ by using such a model.

	\subsubsection{Tuning parameter selection}\label{sec:cross-validation}
	One approach to choosing the tuning parameters $\lambda$ and $\gamma$ is to select the pair of parameters that gives the minimal expected loss, usually estimated by a cross-validation procedure. The candidate parameters can be drawn from a grid satisfying $\lambda>0$ and $\gamma\in[0,1]$. For example, in our simulation study, we perform 5-fold cross-validation using the grid $(\lambda\in\{1, 0.1, \ldots, 1e-10\})\times(\gamma\in\{0,0.1,0.5,0.9,1\})$. The dataset is randomly divided into five folds. For each iteration, the propensity odds models are trained on four folds using penalized optimization and tested on the remaining fold to estimate the out-of-sample tailored loss. The parameter pair that yields the lowest cross-validation error across the five iterations is then used to retrain the model on the entire dataset. This approach enhances the model's ability to generalize to unseen data and ensures a stable estimation of the propensity odds.

	\section{Asymptotic properties}
	\label{s:asymptotic}
	In this section, we investigate the consistency and asymptotic normality of our proposed estimators. Consider the proposed penalization \eqref{penalized-loss} with generic basis functions. By minimizing the penalized average loss $\calLr_\lambda(\alphar)$ for each missing pattern $r\neq 1_d$, we first obtain the minimizers $\hat{\alpha}^r$ and the balancing estimators of propensity odds, $\oddsr(\lob;\hat{\alpha}^r)$. Summing up the propensity odds estimators, we obtain the weights, $\hat{w}(L_i)=\sum_{r\in\calR}\oddsr(\Lobi;\hat{\alpha}^r)$, and a plug-in estimator of $\E\{\psi_\theta(L)\}$: 
	\begin{align*}
		\hat{\bbP}_N\psi_\theta=\frac{1}{N}\sum_{i=1}^N\left\{\mathsf{1}_{R_i=1_d}\hat{w}(L_i)\psi_\theta(L_i)\right\}\ .
	\end{align*} 
	The resulting estimator of $\theta_0$ is the solution to $\hat{\bbP}_N\psi_\theta=0$. Denote it by $\hat{\theta}_N$. Under mild conditions, we show that $\oddsr(\lob;\hat{\alpha}^r)$ is consistent, $\hat{\bbP}_N\psi_\theta$ is asymptotically normal for each $\theta$ in a compact set $\Theta\subset\bbR^q$, and $\hat{\theta}_N$ is consistent and efficient. We require the following set of assumptions to establish the consistency of the proposed estimator.
	
	\begin{assumption}\label{assump2} 
		The following conditions hold for each missing pattern $r\in\calR$:
		\begin{enumerate}[label={\textbf{~\Alph*:}}, ref={Assumption~\theassumption.\Alph*},leftmargin=1cm]
			\item\label{assump-2A}
			There exist constants $0<c_0<C_0$ such that $c_0\le\oddsr(\lob)\le C_0$ for all $\lob\in\domr$ where $\domr\subset\mathbb{R}^{d_r}$ is the domain of $\oddsr$.
			
			\item\label{assump-2B}
			The optimization ${\min_{\alphar}} \E[\calLr\{\oddsr(\Lob;\alphar),R\}]$ has a unique solution $\alphar_0\in\bbR^{K_r}$.
			
			\item\label{assump-2C}
			The total number of basis functions $K_r$ grows as the sample size increases and satisfies $K_r^2=o(N_r)$ where $N_r$ is the number of observations in patterns $r$ and $1_d$. 
			
			\item\label{assump-2D}
			There exist constants $C_1>0$ and $\mu_1>1/2$ such that for any positive integer $K_r$, there exist $\alpha_{K_r}^*\in\bbR^{K_r}$ satisfying $\underset{\lob\in\domr}{\sup}\big|\oddsr(\lob)-\oddsr(\lob;\alpha_{K_r}^*)\big|\le C_1K_r^{-\mu_1}.$
			
			\item\label{assump-2E}
			The Euclidean norm of the basis functions satisfies $\underset{\lob\in\domr}{\sup}\|\Phir(\lob)\|_2=O(K_r^{1/2})$. 
			
			\item\label{assump-2F}
			Let $\lambda_1,\ldots,\lambda_{K_r}$ be the eigenvalues of $\E\{\Phir(\Lob)\Phir(\Lob)\tr\}$ in the non-decreasing order. There exist constants $\lambda_{\min}^*$ and $\lambda_{\max}^*$ such that $0<\lambda_{\min}^*\le\lambda_1\le\lambda_{K_r}\le\lambda_{\max}^*$.
			
			\item\label{assump-2G}
			The tuning parameter $\lambda$ satisfies $\lambda=o(1/\sqrt{K_rN_r})$.
		\end{enumerate}
	\end{assumption}
	\ref{assump-2A} is the boundedness assumption commonly used in missing data and causal inference. It is equivalent to that $P(R=r\mid\lob,R\in\{1_d,r\})$ is strictly bounded away from 0 and 1. The domain $\domr$ is usually assumed to be compact, so it becomes possible to approximate $\oddsr$ with compactly supported functions. \ref{assump-2B} is a standard condition for consistency of minimum loss estimators of $\alphar_0$. It is well known that the uniform approximation error is related to the number of basis functions. Thus, we allow $K_r$ to increase with sample size under certain restrictions in \ref{assump-2C}. The uniform approximation rate $\mu_1$ in \ref{assump-2D} is related to the true propensity odds $\oddsr$ and the choice of basis functions. For instance, the rate $\mu_1=s/d$ for power series and splines if $\oddsr$ is continuously differentiable of order $s$ on $[0,1]^d$ under mild assumptions; see \citet{newey1997convergence} and \citet{fan2022optimal} for details. The restriction $\mu> 1/2$ is a technical condition such that the estimator of propensity odds is consistent. \ref{assump-2E} and \ref{assump-2F} are standard conditions for controlling the magnitude of the basis functions. The Euclidean norm of the basis function vector can increase as the spanned space extends, but its growth rate cannot be too fast. These assumptions are satisfied by many bases such as the regression spline, trigonometric polynomial, and wavelet bases; see, e.g., \citet{newey1997convergence, horowitz2004nonparametric, chen2007large} and \citet{fan2022optimal}. \ref{assump-2G} is a technical assumption of the tuning parameter $\lambda$ for the maintenance of consistency of weights. We now establish the consistency of the estimated odds.
	\begin{theorem}\label{odds}
		Under Assumptions \ref{assump1} and \ref{assump2}, for each missing pattern $r$, we have
		\begin{align*}
			\left\|\oddsr(\ \cdot\ ;\hat{\alpha}^r)-\oddsr\right\|_\infty
			&=O_p\left(\sqrt{\frac{K_r^2}{N_r}}+K_r^{\frac{1}{2}-\mu_1}\right)=o_p(1)\ ,\\
			\left\|\oddsr(\ \cdot\ ;\hat{\alpha}^r)-\oddsr\right\|_{P,2}
			&=O_p\left(\sqrt{\frac{K_r}{N_r}}+K_r^{-\mu_1}\right)=o_p(1)
		\end{align*}
		where $\|X\|_{P,2}^2=\int X^2dP$ is the second moment of a random variable.
	\end{theorem}
	The proof is provided in Appendix \ref{sec:proof-odds}.
	
	Next, we establish the asymptotic normality of the empirical weighted estimating function $\hat{\bbP}_N\psi_\theta$ for each $\theta$. Note that the estimating function $\psi_\theta$ is a $q$-dimensional vector-valued function. We only need to consider each entry separately. Denote the $j$-th entry of $\psi_\theta$ and $u_\theta^r$ as $\psi_{\theta,j}$ and $u_{\theta,j}^r$ respectively. Let $n_{[ \ ]}\{\epsilon,\calF,\cdot\}$ denote the bracketing number of the set $\calF$ by $\epsilon$-brackets with respect to a specific norm. We will need the following additional conditions. 
	\begin{assumption}\label{assump3} 
		The following conditions hold for all missing pattern $r$ and all $\theta\in\Theta$ where $\Theta$ is a compact set:
		\begin{enumerate}[label={\textbf{~\Alph*:}},ref={Assumption~\theassumption.\Alph*},leftmargin=1cm]
			
			\item\label{assump-3A}
			There exist constants $C_2>0$ and $\mu_2>1/2$ such that for any $\theta$ and each missing pattern $r$, there exists a parameter $\beta_\theta^r$ satisfying $\sup_{\lob\in\domr}|u_\theta^r(\lob)-\Phir(\lob)\tr\beta_\theta^r|\le C_2K_r^{-\mu_2}$. 
			
			\item\label{assump-3B}
			Each of the true propensity odds, $\oddsr$, is contained in a set of smooth functions $\calM^r$. There exists constants $C_\calM>0$ and $d_\calM>1/2$ such that $\log n_{[ \ ]}\{\epsilon,\calM^r,L^\infty\}\le C_\calM(1/\epsilon)^{1/d_\calM}$.
			
			\item\label{assump-3C}
			The sets $\Psi:=\{\psi_{\theta,j}:\theta\in\Theta,j=1,\ldots,q\}$ are contained in a function class $\calH$ such that there exists constants $C_\calH>0$ and $d_\calH>1/2$ such that $\log n_{[ \ ]}\{\epsilon,\calH,L_2(P)\}\le C_\calH(1/\epsilon)^{1/d_\calH}$.
			
			\item\label{assump-3D}
			There exists a constant $C_3$ such that for all $j=1,\ldots,q$, 
			\begin{align*}
				\E\left\{\psi_{\theta,j}(L)-u_{\theta,j}^r(\Lob)\right\}^2\le \E\left[\underset{\theta}{\sup}\left\{\psi_{\theta,j}(L)-u_{\theta,j}^r(\Lob)\right\}^2\right]\le C_3^2\ .
			\end{align*}
			
			\item\label{assump-3E}
			$N_r^{1/\{2(\mu_1+\mu_2)\}}=o(K_r)$, which means that the growth rate of the number of basis functions has a lower bound.
		\end{enumerate}
	\end{assumption}
	\ref{assump-3A} is a requirement similar in spirit to \ref{assump-2D} such that the conditional expectation $u_\theta^r(\lob)$ can be well approximated as we extend the space spanned by the basis functions. \ref{assump-3B} and \ref{assump-3C} are conditions on the complexity of the function classes $\calM^r$ and $\calH$ to ensure uniform convergence over $\theta$. These assumptions are satisfied for many function classes. For instance, if $\calM^r$ is a Sobolev class of functions $f:[0,1]\mapsto\bbR$ such that $\|f\|_\infty\le1$ and the ($s-1$)-th derivative is absolutely continuous with $\int(f^{(s)})^2(x)dx\le1$ for some fixed $s\in\mathbb{N}$, then $\log n_{[ \ ]}\{\epsilon,\calM^r,L^\infty)\}\le C(1/\epsilon)^{1/s}$ by Example 19.10 of \citet{van2000asymptotic}. The condition $d_\calM>1/2$ is satisfied for all $s\ge1$. A H\"older class of functions also satisfies this condition \citep{fan2022optimal}. \ref{assump-3D} is a technical condition related to the envelope function such that we can apply the maximal inequality via bracketing. \ref{assump-3E} requires the number of basis functions to grow such that the approximation error decreases in general.
	
	\begin{theorem}\label{psi}
		Suppose that Assumptions \ref{assump1}, \ref{assump2} and \ref{assump3} hold. For any $\theta\in\Theta$,
		\begin{align*}
			\sqrt{N}\left[\hat{\bbP}_N\psi_\theta-\E\{\psi_\theta(L)\}\right]\overset{d}{\to}N(0,V_\theta)
		\end{align*}
		where $V_\theta$ is defined in Theorem \ref{efficiency-bound}.
	\end{theorem}
	To prove the theorem, we utilize a few lemmas of the bracketing number $n_{[ \ ]}\{\epsilon,\calF,\cdot\}$. The proofs of the theorem and lemmas are given separately in Appendices \ref{sec:proof-normality} and \ref{sec:lemma}.
	
	With further assumptions, we show the consistency and asymptotical normality of $\hat{\theta}_N$ that solves $\hat{\bbP}_N\psi_\theta=0$.
	
	\begin{assumption}\label{assump4} 
		The following conditions hold for all missing pattern $r$ and all $\theta\in\Theta$:
		\begin{enumerate}[label={\textbf{~\Alph*:}},ref={Assumption~\theassumption.\Alph*},leftmargin=1cm]
			\item\label{assump-4A}
			For any sequence $\{\theta_n\}\in\Theta$ , $\E\{\underset{1\le j\le q}{\max}|\psi_{\theta_n,j}(L)|\}\to 0$ implies $\|\theta_n-\theta_0\|_2\to 0$
			
			\item\label{assump-4B}
			For each $j$-th entry and any $\delta>0$, there exists an envelop function $f_{\delta,j}$ such that $|\psi_{\theta,j}(l)-\psi_{\theta_0,j}(l)|\le f_{\delta,j}(l)$ for any $\theta$ such that $\|\theta-\theta_0\|_2\le\delta$. Besides, $\|f_{\delta,j}\|_{P,2}\to0$ when $\delta\to0$.
		\end{enumerate}
	\end{assumption}
	\ref{assump-4A} is a standard regularity assumption for Z-estimation. \ref{assump-4B} corresponds to the continuity assumption on $\psi(l,\theta)$ with respect to $\theta$. For example, the Lipschitz class of functions $\{f_\theta:\theta\in\Theta\}$ satisfies this condition if $\Theta$ is compact. More precisely, there exists an uniform envelop function $f$ such that $|f_{\theta_1}(l)-f_{\theta_2}(l)|\le\|\theta_1-\theta_2\|_2f(l)$ for any $\theta_1,\theta_2\in\Theta$ where $\|f\|_{P,2}<\infty$. Now, we establish the theorem.
	
	\begin{theorem}\label{theta}
		Suppose that Assumptions \ref{assump1}--\ref{assump4} hold. Then
		\begin{align*}
			\hat{\theta}_N\xrightarrow{P}\theta_0
		\end{align*}
		and
		\begin{align*}
			N^{\frac{1}{2}}(\hat{\theta}_N-\theta_0)\overset{d}{\to}N(0,D_{\theta_0}^{-1}V_{\theta_0}D_{\theta_0}^{-1\tr})\ ,
		\end{align*}
		where $D_{\theta_0}^{-1}V_{\theta_0}D_{\theta_0}^{-1\tr}$ is the asymptotic variance bound in Theorem \ref{efficiency-bound}. Therefore, $\hat{\theta}_N$ is semiparametrically efficient. 
		
		Similar to \ref{assump-3B}, suppose that the set of derivatives of of estimating functions $\calJ_\Theta:=\{\dot{\psi}_\theta:\theta\in\Theta\}$ satisfies $n_{[ \ ]}\{\epsilon,\calJ_\Theta,L_1(P)\}<\infty$, and $\dot{\psi}_\theta$ is continuous in a neighborhood of $\theta_0$.
		Also suppose that there exists a constant $\lambda_{\max}'$ such that $\sup_{\theta\in\Theta}\|\E\{\dot{\psi}_\theta(L)\dot{\psi}_\theta(L)\tr\}\|_2\le\lambda_{\max}'$.
		Then,
		\begin{align*}
			\hat{D}_{\hat{\theta}_N}=\frac{1}{N}\sum_{i=1}^N\left\{\mathsf{1}_{R_i=1_d}\hat{w}(L_i)\dot{\psi}_{\hat{\theta}_N}(L_i)\right\}
		\end{align*}
		is a consistent estimator of $D_{\theta_0}$. Furthermore, let $\hat{u}_{\hat{\theta}_N}^r$ be an estimator of the conditional expectation $u_{\hat{\theta}_N}^r$ such that $\sup_{\lob\in\domr}|\hat{u}_{\hat{\theta}_N,j}^r(\lob)-u_{\hat{\theta}_N,j}^r(\lob)|=o_p(1)$ for each $j$-th entry. Then, $\hat{D}_{\hat{\theta}_N}^{-1}V_{\hat{\theta}_N}\hat{D}_{\hat{\theta}_N}^{-1\tr}$ is a consistent estimator of $D_{\theta_0}^{-1}V_{\theta_0}D_{\theta_0}^{-1\tr}$ where
		{\footnotesize
			\begin{align}\label{asymptotic-variance}
				\hat{V}_{\hat{\theta}_N}&=\frac{1}{N}\sum_{i=1}^N\left(\hat{F}_i\hat{F}_i\tr\right),\nonumber \\ 
				\hat{F}_i&=\sum_{r\in\calR}\mathsf{1}_{R_i=r}\hat{u}_{\hat{\theta}_N}^r(\Lobi)+\mathsf{1}_{R_i=1_d}\sum_{r\in\calR}\oddsr(\Lobi;\hat{\alpha}^r)\left\{\psi_{\hat{\theta}_N}(L)-\hat{u}_{\hat{\theta}_N}^r(\Lobi)\right\}-\hat{\bbP}_N\psi_{\hat{\theta}_N}\ .
			\end{align}
		}
	\end{theorem}
	The proof is given in Appendix \ref{sec:lemma}.

	\section{Simulation}
	\label{s:simulation}
	A simulation study is conducted to evaluate the finite-sample performance of the proposed estimators. We designed three missing mechanisms that satisfy the CCMV condition. For each setting, we simulated 1,000 independent data sets, each of size $N$=1,000, where $X_j,\ j=1,2,3$, are generated independently from a truncated standard normal distribution with support $[-3,3]$. We considered a logistic regression model $\mathrm{logit}\{P(Y=1\mid X)\}=\theta_1X_1+\theta_2X_2+\theta_3X_3+\theta_4$ where the true coefficients $\theta_0=(1,-1,1,-2)$ are the parameters of interest. The response variable $Y$ is always observed, while $X_1,X_2, X_3$ could be missing. Four non-monotone response patterns are considered: $L^{1111}=(Y,X_1,X_2,X_3)$, $L^{1110}=(Y,X_1,X_2)$, $L^{1101}=(Y,X_1,X_3)$ and $L^{1100}=(Y,X_1)$. The categorical variable $R\in\{1111,1110,1101,1100\}$ indicating the response patterns is generated from a multinomial distribution with the probabilities shown in Appendix \ref{sec:setting}. 
	
	\begin{table}
		\caption{Results of the simulation study based on 1000 replications under three CCMV conditions.}
		\label{simulation-result}
		\begin{center}
			\begin{tabular}{lcccccccc}
				\toprule
				\multicolumn{1}{l}{Method} & \multicolumn{4}{c}{Bias} & \multicolumn{4}{c}{MSE}\\ 
				\cmidrule(lr){2-5}
				\cmidrule(lr){6-9}
				& $\theta_1$ & $\theta_2$ & $\theta_3$ & $\theta_4$ 
				& $\theta_1$ & $\theta_2$ & $\theta_3$ & $\theta_4$ \\ 
				\hline
				\multicolumn{1}{c}{Setting 1}\\
				Full & -0.022 & 0.007 & -0.010 & 0.009 & 0.016 & 0.012 & 0.012 & 0.012 \\ 
				Complete-case & 0.512 & 0.239 & 0.002 & 0.103 & 0.297 & 0.100 & 0.031 & 0.048 \\ 
				Mean-score & 0.218 & 0.098 & -0.004 & 0.084 & 0.077 & 0.050 & 0.034 & 0.047 \\ 
				True-weight & -0.065 & 0.086 & -0.044 & 0.058 & 0.049 & 0.094 & 0.059 & 0.072 \\ 
				Entropy-linear & -0.067 & 0.082 & -0.045 & 0.057 & 0.038 & 0.091 & 0.059 & 0.071 \\ 
				Entropy-parametric & -0.067 & 0.082 & -0.045 & 0.057 & 0.038 & 0.091 & 0.059 & 0.071 \\ 
				Entropy-basis & -0.039 & 0.064 & -0.027 & 0.088 & 0.030 & 0.053 & 0.044 & 0.054 \\ 
				Proposed & -0.024 & 0.060 & -0.032 & 0.112 & 0.033 & 0.047 & 0.043 & 0.054 \\ 
				\hline
				\multicolumn{1}{c}{Setting 2}\\
				Full & -0.022 & 0.007 & -0.010 & 0.009 & 0.016 & 0.012 & 0.012 & 0.012 \\ 
				Complete-case & 0.943 & 0.310 & -0.176 & -0.089 & 0.921 & 0.152 & 0.070 & 0.043 \\ 
				Mean-score & 0.318 & 0.085 & 0.059 & 0.052 & 0.140 & 0.073 & 0.045 & 0.041 \\ 
				True-weight & -0.045 & 0.080 & -0.063 & 0.038 & 0.050 & 0.100 & 0.067 & 0.065 \\ 
				Entropy-linear & -0.001 & 0.126 & -0.111 & 0.051 & 0.041 & 0.118 & 0.081 & 0.072 \\ 
				Entropy-parametric & -0.056 & 0.088 & -0.065 & 0.044 & 0.042 & 0.092 & 0.065 & 0.065 \\ 
				Entropy-basis & -0.062 & 0.054 & -0.058 & 0.070 & 0.042 & 0.074 & 0.064 & 0.067 \\ 
				Proposed & -0.023 & 0.029 & -0.044 & 0.089 & 0.029 & 0.051 & 0.055 & 0.057 \\ 
				\hline
				\multicolumn{1}{c}{Setting 3}\\
				Full & -0.022 & 0.007 & -0.010 & 0.009 & 0.016 & 0.012 & 0.012 & 0.012 \\ 
				Complete-case & 0.805 & 0.889 & -0.043 & 0.252 & 0.689 & 0.892 & 0.041 & 0.111 \\ 
				Mean-score & 0.253 & 0.119 & 0.002 & -0.071 & 0.098 & 0.090 & 0.050 & 0.061 \\ 
				True-weight & -0.076 & 0.195 & -0.115 & 0.081 & 0.115 & 0.220 & 0.134 & 0.129 \\ 
				Entropy-linear & -0.252 & 0.349 & -0.153 & -0.022 & 0.135 & 0.295 & 0.149 & 0.096 \\ 
				Entropy-parametric & -0.158 & 0.265 & -0.136 & 0.104 & 0.098 & 0.218 & 0.151 & 0.143 \\ 
				Entropy-basis & -0.084 & 0.187 & -0.129 & 0.109 & 0.135 & 0.244 & 0.147 & 0.145 \\ 
				Proposed & 0.016 & 0.035 & -0.021 & 0.047 & 0.039 & 0.054 & 0.052 & 0.071 \\ 
				\bottomrule 
			\end{tabular}
		\end{center}
	\end{table}
	We first analyzed the simulated data with the full dataset (Full), which is the ideal case with no missingness. We then analyzed the data in the complete case pattern (Complete-case), for which data in all missing patterns $r\neq 1_d$ are discarded, and an unweighted analysis is used for the remaining data. As we described before, the mean-score approach typically requires the discretization of continuous variables. To implement the mean-score method in \citet{chatterjee2010inference} (Mean-score), we created discrete surrogates as $\tilde{X}_j=\mathsf{1}_{X_j\ge0},\ j=1,2,3$. Note that the mean-score method makes a missing-at-random (MAR) assumption. Next, we considered the inverse propensity weighting methods with the true inverse propensity weights (True-weight). We also examined the performance of the estimators based on the estimated propensity odds using the entropy loss. \citet{tchetgen2018discrete} modeled each propensity odds using linear logistic regression (Entropy-linear), \emph{i.e.}, the logarithm of propensity odds, $\log(\oddsr)$, is a linear function of the shared variables in patterns $r$ and $1_d$, which is correctly specified only for setting 1. For settings 2 and 3, we also applied a parametric logistic regression with the correct model (Entropy-parametric). Note that they are unpenalized methods. When implementing the proposed methodology (Proposed), we set the basis functions $\Phir$ as the tensor products of functions of variables shared in patterns $r$ and $1_d$, where polynomials of degrees up to three are chosen for each continuous variable, and a binary indicator function is chosen for discrete response $Y$. Then, we orthogonalized the basis functions and followed the choices of $\{t_k\}_{k=1}^{K_ r}$ and matrix $\bfDr$ in \ref{s:tuning} to construct the penalty function. Notice that the constant and linear functions have zero roughness. In practice, to achieve stable estimations of the propensity odds, we set the tolerance to the empirical imbalances of these functions to be the same as the smallest positive roughness of other basis functions. We also applied the above basis functions and penalties with entropy loss (Entropy-basis) to estimate the propensity odds for comparing tailored and entropy loss functions. The biases and mean squared errors of each coefficient are shown in Table~\ref{simulation-result}. 
	
	\begin{table}[h]
		\caption{Results of the asymptotic property study: Avg Est SD stands for the average of estimated asymptotic standard deviation. MC SD is short for the Monte Carlo standard deviation, which is the standard deviation of proposed estimators of $\theta$ over 1000 replications. 95\% confidence interval was constructed using estimators of $\theta$ and asymptotic variance for each simulated dataset. The confidence interval coverage probabilities are reported based on 1000 replications.} 
		\label{variance-CI}
		\begin{center}
			\begin{tabular}{llcccccccc}
				\toprule
				\multicolumn{2}{l}{} & \multicolumn{4}{c}{Avg Est SD/MC SD} & \multicolumn{4}{c}{Coverage probabilities}\\ 
				\cmidrule(lr){3-6}
				\cmidrule(lr){7-10}
				& Sample size & $\theta_1$ & $\theta_2$ & $\theta_3$ & $\theta_4$ & $\theta_1$ & $\theta_2$ & $\theta_3$ & $\theta_4$\\
				\hline
				Setting 1 & N=1000 & 0.937 & 0.861 & 0.961 & 0.970 & 0.938 & 0.908 & 0.930 & 0.926 \\ 
				& N=2000 & 0.939 & 0.979 & 0.990 & 1.032 & 0.937 & 0.942 & 0.940 & 0.934 \\ 
				& N=5000 & 1.023 & 1.091 & 1.040 & 1.164 & 0.947 & 0.958 & 0.957 & 0.962 \\ 
				& N=10000 & 1.079 & 1.115 & 1.039 & 1.219 & 0.960 & 0.959 & 0.954 & 0.962 \\ 
				Setting 2 & N=1000 & 0.978 & 0.843 & 0.892 & 0.958 & 0.946 & 0.897 & 0.930 & 0.934 \\ 
				& N=2000 & 0.926 & 0.952 & 0.916 & 0.950 & 0.930 & 0.931 & 0.931 & 0.928 \\ 
				& N=5000 & 0.967 & 0.971 & 0.933 & 1.037 & 0.944 & 0.919 & 0.930 & 0.946 \\ 
				& N=10000 & 1.061 & 1.097 & 0.984 & 1.041 & 0.959 & 0.962 & 0.951 & 0.946 \\ 
				Setting 3 & N=1000 & 0.877 & 0.793 & 0.856 & 0.856 & 0.899 & 0.893 & 0.904 & 0.916 \\ 
				& N=2000 & 0.918 & 0.932 & 0.925 & 0.905 & 0.921 & 0.924 & 0.940 & 0.921 \\ 
				& N=5000 & 0.952 & 1.018 & 1.007 & 0.947 & 0.936 & 0.951 & 0.959 & 0.939 \\ 
				& N=10000 & 0.992 & 1.051 & 1.012 & 1.010 & 0.946 & 0.959 & 0.950 & 0.947 \\ 
				\bottomrule
			\end{tabular}
		\end{center}
	\end{table}
	
	In all three settings, the mean-score method has bad performances (considerable bias over $\theta_1$) that are likely attributed to the misspecified missing mechanisms and discretization. In settings 1 and 2, compared with the linear model and correctly specified parametric model, using basis functions as regressors and adding penalization alleviates overfitting such that the inverse propensity weights and estimates are more stable. The proposed method provides comparable or even smaller errors in these two settings. It is expected since the tailored loss also encourages the balance of observed variables. In setting 3, all the methods misspecify the missing mechanism, but the proposed method outperforms the others.
	
	We also assess the performance of asymptotic variance estimators \eqref{asymptotic-variance} and confidence interval coverage using the estimated standard errors. The results in Table~\ref{variance-CI} show that the average of estimated asymptotic standard deviations is close to the Monte Carlo standard deviations, which are the standard deviations of estimated parameters from 1000 datasets. The $95\%$ confidence interval constructed by estimated asymptotic variance has close to nominal coverage, especially when the sample size increases.

	\section{Real data application}
	We apply the proposed method to a study of hip fracture risk factors among male veterans \citep{barengolts2001risk}. The binary outcome of interest was the presence $(Y=1)$ or absence $(Y=0)$ of hip fracture. Preliminary exploratory analyses suggested that nine risk factors are potentially significant. A detailed description of the data set and its missing data patterns is in \citet{chen2004nonparametric}. Three (BMI, HGB, and Albumin) of these factors are continuous variables, and the remaining six (Etoh, Smoke, Dement, Antiseiz, LevoT4, Antichol) are discrete variables. Including $1_d$, where all variables are observed, there is a total of 38 missingness patterns. The parameters of interest are the coefficients of a logistic regression model with these nine variables as the predictors and the presence or absence of hip fracture as the binary outcome. 
	
	We implemented the same Complete-case analysis as in \citet{chen2004nonparametric} and our proposed method under the CCMV assumption. Like the simulation study, we generated orthogonalized tensor product basis functions for continuous variables and considered two penalties constructed by the roughness norm. Indicator functions are created for each discrete variable as basis functions. Since there are many combinations of possible values of discrete variables, we do not consider any further tensor products. In other words, the indicator functions of discrete variables are added to the log-linear model for propensity odds in an additive way. For example, if $\Lob=$(BMI, HGB, Dement, Antiseiz), the propensity odds model is
	\begin{align*}
		\log\oddsr(\lob)=\sum_{k=1}^{16}\alphar_k\phir_k(\mathrm{BMI, HGB})+\alphar_{17}\mathsf{1}_{\mathrm{Dement}=1}+\alphar_{18}\mathsf{1}_{\mathrm{Antiseiz}=1}
	\end{align*}
	where $\{\phir_k\}_{k=1}^{16}$ are tensor products of two sets of functions, each of which contains polynomials up to degree 3 for one continuous variable. Similarly, we should penalize the indicator functions that perform similarly to lower-order moments for stabilization. The relative tolerance to the empirical imbalances of the indicator functions was chosen to be the same as that for the constant and linear functions. 
	
	We present the parameter estimates, standard errors, and p-values in Table~\ref{realdata}. The results show that our proposed estimators have standard errors similar to those from the complete-case analysis. All nine predictor variables are statistically significant at the $5\%$ level. Our results largely agree with those of \citet{chen2004nonparametric}. One apparent difference is that the coefficient of LevoT4 is significant in our analysis but not in \citet{chen2004nonparametric}. The main difference between the two estimators is the missing mechanism assumptions.
	
	One may notice that some patterns may have rare observations. There are 26 patterns having at most three observations. It may draw attention to the stability of the propensity odds estimator for those patterns. Through the definition, the true propensity odds between a rare pattern and the complete cases should be negligible. Therefore, the contribution of rare patterns to the combined weights should also be negligible. Although the sample size may be too small to guarantee the asymptotic properties, the estimated odds for rare patterns using the proposed penalized estimation only slightly contribute to the total weights. Thus, the total weights and the weighted estimator for the parameters of interest remain stable.
	
	\begin{table}[h]
		\caption{Results of the Hip Fracture Data analysis: The asymptotic variance estimator is used to calculate the standard error(SE) and p-value.} 
		\label{realdata}
		\begin{center}
			\begin{tabular}{lcccccc}
				\toprule
				\multicolumn{1}{l}{Parameters} & \multicolumn{3}{c}{Complete-case} & \multicolumn{3}{c}{Proposed}\\ 
				\cmidrule(lr){2-4}
				\cmidrule(lr){5-7}
				& Estimate & SE & p-value & Estimate & SE & p-value \\ 
				\hline
				Etoh & 1.391 & 0.391 & $<$1e-3 & 1.410 & 0.397 & $<$1e-3 \\ 
				Smoke & 0.929 & 0.400 & 0.020 & 0.972 & 0.386 & 0.012 \\ 
				Dementia & 2.509 & 0.724 & 0.001 & 2.595 & 0.689 & $<$1e-3 \\ 
				Antiseiz & 3.311 & 1.064 & 0.002 & 3.435 & 0.935 & $<$1e-3 \\ 
				LevoT4 & 2.010 & 1.015 & 0.048 & 1.941 & 0.771 & 0.012 \\ 
				AntiChol & -1.918 & 0.768 & 0.012 & -1.872 & 0.796 & 0.019 \\ 
				BMI & -0.104 & 0.039 & 0.007 & -0.100 & 0.040 & 0.013 \\ 
				logHGB & -2.597 & 1.202 & 0.031 & -2.541 & 1.232 & 0.039 \\ 
				Albumin & -0.911 & 0.353 & 0.010 & -0.857 & 0.378 & 0.024 \\ 
				\bottomrule
			\end{tabular}
		\end{center}
	\end{table}

	\section{Discussion}
	We studied the estimation of model parameters defined by estimating functions with non-monotone missing data under the complete-case missing variable condition, which is a case of missing not at random. Using tailored loss for balance, functional bases for flexible modeling and penalization for stable estimation, we propose a method that can be viewed as a generalization of both inverse probability weighting and mean-score approach. The proposed framework improves the reliability of inferences and has significant implications for the analysis of datasets with non-monotone missingness, which is frequently encountered in biomedical data. Given the limited availability of analytical tools for addressing non-monotone missingness, we hope to fill an important gap in the literature.
	
	Despite its strengths, the proposed method has certain limitations. The CCMV assumption, while plausible, is inherently unverifiable from observed data, which may limit its acceptance in some applications. Recently, \cite{chen2022pattern} generalizes the assumption into pattern graphs, which allows detailed relationships among non-monotone missing patterns to be assumed. 
	One can notice that the CCMV assumption is a special case of the pattern graph. An example of such a graph is shown in Figure \ref{fig:CCMV}. All the missing patterns have the same and only parent, the fully observed pattern. 
	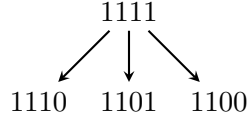
\begin{figure}[h]
		\centering
		\begin{tikzpicture}[node distance=1.2cm]
			\node(1d){1111};
			\node[below of = 1d,xshift=-1.2cm](r1){1110};
			\node[below of = 1d,xshift=0cm](r2){1101};
			\node[below of = 1d,xshift=1.2cm](r3){1100};
			\path[line](1d)--(r1);
			\path[line](1d)--(r2);
			\path[line](1d)--(r3);
		\end{tikzpicture}
		\caption{The pattern graph for CCMV}\label{fig:CCMV}
	\end{figure}
	A future research direction is to extend our tailored loss estimation to allow more flexible identifying assumptions, which are represented by general pattern graphs. The proposed method is well-suited to serve as a tool for these assumptions. Another important direction is to investigate the impact of misspecification on the identifying assumptions and explore ways to incorporate robustness checks. This includes designing sensitivity analyses that account for potential deviations from CCMV or other assumption structures.

	
	\section*{Acknowledgements}
	This work is partly supported by the National Science Foundation (DMS-1711952).
	Portions of this research were conducted with the advanced computing resources provided by Texas A\&M High Performance Research Computing.

	\begin{appendices}
		\section{Details of simulation settings}
		\label{sec:setting}
		The probabilities of data belonging to the four response patterns are respectively:
		\begin{align*}
			P(R=1111\mid L)&=\frac{1}{1+\odds^{1110}(L^{1110})+\odds^{1101}(L^{1101})+\odds^{1100}(L^{1100})},\\ 
			P(R=1110\mid L)&=\frac{\odds^{1110}(L^{1110})}{1+\odds^{1110}(L^{1110})+\odds^{1101}(L^{1101})+\odds^{1100}(L^{1100})},\\ 
			P(R=1101\mid L)&=\frac{\odds^{1101}(L^{1101})}{1+\odds^{1110}(L^{1110})+\odds^{1101}(L^{1101})+\odds^{1100}(L^{1100})},\\ 
			P(R=1100\mid L)&=\frac{\odds^{1100}(L^{1100})}{1+\odds^{1110}(L^{1110})+\odds^{1101}(L^{1101})+\odds^{1100}(L^{1100})},
		\end{align*}
		where $\oddsr(\Lob)=P(R=r\mid\Lob)/P(R=1111\mid\Lob)$ is the true propensity odds, which only depends on the shared variables $\Lob$ under CCMV condition, for $r=1110,1101,1100$. 
		
		\subsection*{Setting 1}
		The logarithms of true propensity odds are:
		\begin{align*}
			\log\odds^{1110}(L^{1110})=&X_1+X_2-\mathsf{1}_{Y=1}-0.5\ ,\\ 
			\log\odds^{1101}(L^{1101})=&\frac{1}{2}X_1+X_3+-0.5\mathsf{1}_{Y=1}-0.3\ ,\\
			\log\odds^{1100}(L^{1100})=&\frac{3}{2}X_1-\mathsf{1}_{Y=1}-0.4\ .
		\end{align*}
		
		\subsection*{Setting 2}
		The logarithms of true propensity odds are:
		\begin{align*}
			\log\odds^{1110}(L^{1110})=&\frac{1}{5}(X_1^2-9)(X_1+1.5)+\frac{1}{5}(X_2^2-9)(X_2+1)\\
			&+\frac{1}{10}(X_1+2)(X_2+2)(X_2-1)-2\mathsf{1}_{Y=1}+3\ ,\\
			\log\odds^{1101}(L^{1101})=&-\frac{1}{5}(X_1^2-9)(X_1+1)+\frac{1}{5}(X_3^2-9)(X_3+1.5)-2\mathsf{1}_{Y=1}\ ,\\
			\log\odds^{1100}(L^{1100})=&-\frac{1}{5}(X_1+2)(X_1+0.5)(X_1-4)-2\mathsf{1}_{Y=1}-1\ .
		\end{align*}
		
		\subsection*{Setting 3} 
		The logarithms of true propensity odds are:
		\begin{align*}
			\log\odds^{1110}(L^{1110})=&\frac{1}{10}(X_1^2-9)(X_1^2-4)X_1+\frac{1}{10}(X_2^2-9)(X_2+1)\\
			&+\frac{1}{4}(X_1+2)(X_2+2)(X_2-1)-2\mathsf{1}_{Y=1}\ ,\\
			\log\odds^{1101}(L^{1101})=&\frac{1}{10}(X_1^2-9)(X_1+1)+\frac{1}{10}(X_3^2-9)(X_3^2-4)X_3\\
			&+\mathsf{1}_{Y=1}\left\{(X_1+1)(X_3^2-4)-2\right\}\ ,\\
			\log\odds^{1100}(L^{1100})=&\mathsf{1}_{Y=0}\left\{\frac{1}{5}(X_1^2-9)(X_1^2-4)X_1-1\right\}\\
			&-\mathsf{1}_{Y=1}\left\{\frac{1}{10}(X_1^2-9)(X_1^2-2.5^2)(X_1+0.5)\right\}\ .
		\end{align*}

		We assume that the number of basis $K_r$ increases with sample size $N$. Precisely, one should predetermine a way to choose $K_r$ depending on the size of the fully observed pattern and missing pattern $r$, $N_r$, which also increases with $N$. In the following proofs, we use ``a large enough $N$'' to describe the requirements for sample size, while the explicit conditions may be related to $K_r$ or $N_r$.

		\section{Proof of Theorem \ref{efficiency-bound}}
		\label{sec:proof-efficiency}
		
		{\bf Proof sketch:}
		To show $D_{\theta_0}^{-1}V_{\theta_0}D_{\theta_0}^{-1\tr}$ is the efficiency bound, we closely follow the structure of semiparametric efficiency bound derivation of \cite{newey1990semiparametric}, \cite{bickel1993efficient} and \cite{chen2008semiparametric}. Briefly speaking, the efficiency bound is defined as the supremum of the Cramer-Rao bounds for all regular parametric submodels. If it exists, the efficient estimator should have the efficient influence function that gives the minimum $\E\{\zeta\zeta\tr\}$, which is known as the asymptotic variance of such an estimator. Based on Theorem 2.2 in \citet{newey1990semiparametric}, the regular and asymptotically linear estimator has the influence function satisfying \eqref{path-differentiability}. 
		Theorem 3.1 in \citet{newey1990semiparametric} reveals that the efficient influence function is in the tangent space $\calT$ where $\calT$ is defined as the mean square closure of all $q$-dimensional linear combinations of scores for all smooth parametric submodels.

		\begin{proof}
			Consider an arbitrary parametric submodel for the likelihood function of $(L,R)$
			\begin{align*}
				L_\alpha(l,r)
				=\prod_{s\in\calR}f_\alpha(l^s,s)^{\mathsf{1}_{r=s}}
			\end{align*}
			where $f_{\alpha_0}(l^s,s)$ gives the true distribution of observed variables in pattern $s$. The resulting score function is given by
			\begin{align*}
				S_\alpha(l,r)
				=\sum_{s\in\calR}\mathsf{1}_{r=s}S_\alpha(l^s,s)
			\end{align*}
			where $S_\alpha(l^{s},s)=\partial\log f_\alpha(l^{s},s)/\partial\alpha$ for $s\in\calR$. Recall that the parameter of interest $\theta_0$ is the solution to $\E\{\psi_\theta(L)\}=0$ and thus is a function of $\alpha$. Pathwise differentiability follows if we can find an influence function $\zeta(L,R)$ for all regular parametric submodels such that
			\begin{align}\label{path-differentiability}
				\frac{\partial\theta_0(\alpha_0)}{\partial\alpha}=\E\{\zeta(L,R)S_{\alpha_0}(L,R)\}\ .
			\end{align}
			To save notations, $\alpha$ is also used as the true parameter value $\alpha_0$ when the context is clear. Chain rule and Leibniz integral rule (differentiating under the integral) gives
			\begin{align*}
				\frac{\partial\E\{\psi_\theta(L)\}}{\partial\alpha}
				&=\int\frac{\partial\psi_\theta(l)f_\alpha(l)}{\partial\alpha}dl
				=\int\left\{\frac{\partial\psi_\theta(l)}{\partial\theta}\frac{\partial\theta(\alpha)}{\partial\alpha}f_\alpha(l)
				+\psi_\theta(l)\frac{\partial f_\alpha(l)}{\partial\alpha}\right\}dl\\
				&=\frac{\partial\theta(\alpha)}{\partial\alpha}\int\frac{\partial\psi_\theta(l)}{\partial\theta}f_\alpha(l)dl
				+\int\psi_\theta(l)\frac{\partial\log f_\alpha(l)}{\partial\alpha}f_\alpha(l)dl\\
				&=\frac{\partial\theta(\alpha)}{\partial\alpha}\frac{\partial\E\{\psi_\theta(L)\}}{\partial\theta}
				+\E\left\{\psi_\theta(L)\frac{\partial\log f_\alpha(L)}{\partial\alpha}\right\}\ .
			\end{align*}
			Therefore,
			\begin{align*}
				\frac{\partial\theta_0(\alpha_0)}{\partial\alpha}
				=-\left[\left.\frac{\partial\E\{\psi_\theta(L)\}}{\partial\theta}\right\vert_{\theta_0}\right]^{-1}\E\left\{\psi_{\theta_0}(L)\frac{\partial\log f_\alpha(L)}{\partial\alpha}\right\}\ .
			\end{align*}
			By the identification assumption, $f_\alpha(\lms\mid\lob,r)=f_\alpha(\lms\mid\lob,1_d)$. Thus, the marginal density
			\begin{align*}
				f_\alpha(l)
				=f_\alpha(l,1_d)+\sum_{r\neq 1_d, r\in\calR}f_\alpha(\lob,r)f_\alpha(\lms\mid\lob,1_d)\ .
			\end{align*}
			Define $S_\alpha(\lms\mid\lob,1_d)=\partial\log f_\alpha(\lms\mid\lob,1_d)/\partial\alpha$ where $\int S_\alpha(\lms\mid\lob,1_d)f_\alpha(\lms\mid\lob,1_d)d\lms=0$. Then,
			\begin{align*}
				\E\left\{\psi_\theta(L)\frac{\partial\log f_\alpha(L)}{\partial\alpha}\right\}
				&=\int\psi_\theta(l)\frac{\partial f_\alpha(l,1_d)}{\partial\alpha}dl\\
				&\quads+\sum_{r\neq 1_d, r\in\calR}\int\psi_\theta(l)\frac{\partial f_\alpha(\lob,r)f_\alpha(\lms\mid\lob,1_d)}{\partial\alpha}dl\ .
			\end{align*}
			The first term is 
			\begin{align*}
				\int\psi_\theta(l)\frac{\partial f_\alpha(l,1_d)}{\partial\alpha}dl
				=\int\psi_\theta(l)S_\alpha(l,1_d)f_\alpha(l,1_d)dl
				=\E\{\mathsf{1}_{R=1_d}\psi_\theta(L)S_\alpha(L,1_d)\}\ .
			\end{align*}
			For each $r\neq 1_d, r\in\calR$,
			\begin{align*}
				&\quads\int\psi_\theta(l)\frac{\partial f_\alpha(\lob,r)f_\alpha(\lms\mid\lob,1_d)}{\partial\alpha}dl\\
				&=\int\psi_\theta(l)\left\{\frac{\partial f_\alpha(\lob,r)}{\partial\alpha}f_\alpha(\lms\mid\lob,1_d)+f_\alpha(\lob,r)\frac{\partial f_\alpha(\lms\mid\lob,1_d)}{\partial\alpha}\right\}dl\\
				&=\int\psi_\theta(l)S_\alpha(\lob,r)f_\alpha(\lob,r)f_\alpha(\lms\mid\lob,1_d)dl\\
				&\quads+\int\psi_\theta(l)S_\alpha(\lms\mid\lob,1_d)f_\alpha(\lob,r)f_\alpha(\lms\mid\lob,1_d)dl\ .
			\end{align*}
			The first integral is
			\begin{align*}
				&\quads\int\psi_\theta(l)S_\alpha(\lob,r)f_\alpha(\lob,r)f_\alpha(\lms\mid\lob,1_d)dl\\
				&=\int\left\{\int\psi_\theta(l)f_\alpha(\lms\mid\lob,1_d)d\lms\right\}S_\alpha(\lob,r)f_\alpha(\lob,r)dl\\
				&=\int\E\{\psi_\theta(L)\mid\Lob=\lob,R=1_d\}S_\alpha(\lob,r)f_\alpha(\lob,r)dl\\
				&=\E\left\{\mathsf{1}_{R=r}u_\theta^r(\Lob)S_\alpha(L^r,r)\right\}\ .
			\end{align*}
			Recall that $\oddsr(\lob)=P(R=r\mid\lob)/P(R=1_d\mid\lob)=p(\lob,r)/p(\lob,1_d)$. The second integral is
			\begin{align*}
				&\quads\int\psi_\theta(l)S_\alpha(\lms\mid\lob,1_d)f_\alpha(\lob,r)f_\alpha(\lms\mid\lob,1_d)dl\\
				&=\int\psi_\theta(l)S_\alpha(\lms\mid\lob,1_d)\frac{f_\alpha(\lob,r)}{f_\alpha(\lob,1_d)}f_\alpha(\lob,1_d)f_\alpha(\lms\mid\lob,1_d)dl\\
				&=\E\left\{\mathsf{1}_{R=1_d}\psi_\theta(L)\oddsr(\Lob)S_\alpha(\Lms\mid\Lob,1_d)\right\}\\
				&=\E\left[\mathsf{1}_{R=1_d}\left\{\psi_\theta(L)-u_\theta(\Lob)\right\}\oddsr(\Lob)S_\alpha(\Lms\mid\Lob,1_d)\right]\\
				&=\E\left[\mathsf{1}_{R=1_d}\left\{\psi_\theta(L)-u_\theta(\Lob)\right\}\oddsr(\Lob)S_\alpha(L,1_d)\right]
			\end{align*}
			since
			\begin{align*}
				&\quads\E\left\{\mathsf{1}_{R=1_d}u_\theta(\Lob)\oddsr(\Lob)S_\alpha(\Lms\mid\Lob,1_d)\right\}\\
				&=\E\left[\mathsf{1}_{R=1_d}u_\theta(\Lob)\oddsr(\Lob)\E\left\{S_\alpha(\Lms\mid\Lob,1_d)\mid\Lob,R=1_d\right\}\right]
				=0\ ,
			\end{align*}
			and
			\begin{align*}
				&\quads\E\left[\E\left\{\psi_\theta(L)\mid\Lob,R=1_d\right\}\mathsf{1}_{R=1_d}\oddsr(\Lob)S_\alpha(L^r,1_d)\right]\\
				&=\E\left\{\mathsf{1}_{R=1_d}u_\theta(\Lob)\oddsr(\Lob)S_\alpha(L^r,1_d)\right\}
			\end{align*}
			where $S_\alpha(\lob,1_d)=\partial\log f_\alpha(\lob,1_d)/\partial\alpha$ and $S_\alpha(l,1_d)=S_\alpha(\lob,1_d)+S_\alpha(\lms\mid\lob,1_d)$. Therefore,
			{\footnotesize
				\begin{align*}
					\E\left\{\psi_\theta(L)\frac{\partial\log f_\alpha(L)}{\partial\alpha}\right\}
					&=\E\{\mathsf{1}_{R=1_d}\psi_\theta(L)S_\alpha(L,1_d)\}
					+\sum_{r\neq 1_d, r\in\calR}\E\left\{\mathsf{1}_{R=r}u_\theta^r(\Lob)S_\alpha(L^r,r)\right\}\\
					&\quads+\sum_{r\neq 1_d, r\in\calR}\E\left[\mathsf{1}_{R=1_d}\left\{\psi_\theta(L)-u_\theta(\Lob)\right\}\oddsr(\Lob)S_\alpha(L,1_d)\right]\ .
				\end{align*}
			}
			Notice that $\E\{\psi_{\theta_0}(L)\}=0$. Recall that we denote the derivative $\partial\E\{\psi_\theta(L)\}/\partial\theta$ at $\theta_0$ as $D_{\theta_0}$. Let the influence function $\zeta(L,R)=
			-D_{\theta_0}^{-1}F_{\theta_0}(L,R)$ where
			{\small
				\begin{align*}
					F_\theta(L,R)=\mathsf{1}_{R=1_d}\sum_{r\in\calR}\oddsr(\Lob)\{\psi_\theta(L)-u_\theta^r(\Lob)\}+\sum_{r\in\calR}\mathsf{1}_{R=r}u_\theta^r(\Lob)-\E\{\psi_\theta(L)\}\ .
				\end{align*}
			}
			Thus, equation \eqref{path-differentiability} is satisfied, and $\theta(\alpha)$ is pathwise differentiable. To calculate the efficiency bound, we need to find the efficient influence function, which is the projection of $\zeta(L,R)$ on the tangent set. The tangent set $\calT$ is defined as the mean square closure of all $q$-dimensional linear combinations of scores $S_\alpha$ for smooth parametric submodels,
			\begin{align*}
				\calT=\left\{h(L,R)\in\bbR^q:\E\{\|h\|^2\}\le\infty,
				\exists A_jS_{\alpha j}\textrm{ with }\lim_{j\to\infty}\E\{\|h-A_jS_{\alpha j}\|\}=0
				\right\}
			\end{align*}
			where $A_j$ is a constant matrix with $q$ rows. It can be verified by similar arguments as in \cite{newey1990semiparametric}. Note that each score $S_{\alpha j}$ has the same form of decomposition, which is the summation of partial scores for each missing pattern $r$. It is also easy to see that $\calT$ is linear and $\zeta$ belongs to the tangent set where $\mathsf{1}_{R=r}u_{\theta_0}^r(\Lob)$ taking the role of $\mathsf{1}_{R=r}S_\alpha(\lob,r)$ for $r\neq 1_d, r\in\calR$ and the rest taking the role of $\mathsf{1}_{R=1_d}S_\alpha(l,1_d)$. Therefore all the conditions of Theorem 3.1 in \cite{newey1990semiparametric} hold, so the efficiency bound for regular estimators of the parameter $\theta$ is given by $D_{\theta_0}^{-1}V_{\theta_0}D_{\theta_0}^{-1\tr}$ where $V_{\theta_0}=\E\{F_{\theta_0}(L,R)F_{\theta_0}(L,R)\tr\}$.
		\end{proof}

		\section{Proof of Theorem \ref{odds}}
		\label{sec:proof-odds}

		{\bf Proof sketch:}
		\ref{assump-2D} assumes that there is a close approximation of the true propensity odds. We will show our estimator is close to the approximation. With the help of a few inequalities, the distance of functions is proportionally bounded by the distance of coefficients. The key to the proof is to show the distance between two coefficients converges with a certain order. The problem is converted to the study of a quadratic form with random coefficients in Lemma \ref{alpha}. The quadratic coefficients form a symmetric random matrix. By the Weyl inequality, we can connect the random matrix with the magnitude of basis functions. So, we can apply the matrix Bernstein inequality to provide bounds on the spectral norm, \emph{i.e.}, the largest eigenvalue. Similarly, we can show that the linear coefficients are also bounded. Lemmas \ref{quadratic} and \ref{linear} provide the bound for the quadratic and linear coefficients respectively.

		\begin{proof} 
			By the triangle inequality and \ref{assump-2D},
			{\small
				\begin{align*}
					&\quads\underset{\lob\in\domr}{\sup}\left| \oddsr(\lob;\hat{\alpha}^r)-\oddsr(\lob)\right|\\
					&\le\underset{\lob\in\domr}{\sup}\left| \oddsr(\lob;\hat{\alpha}^r)-\oddsr(\lob;\alpha_{K_r}^*)\right|
					+\underset{\lob\in\domr}{\sup}\left| \oddsr(\lob;\alpha_{K_r}^*)-\oddsr(\lob)\right|\\
					&\le\underset{\lob\in\domr}{\sup}\left|\exp\left\{\Phir(\lob)\tr\hat{\alpha}^r\right\}-\exp\left\{\Phir(\lob)\tr\alpha_{K_r}^*\right\}\right|
					+C_1K_r^{-\mu_1}\ .
				\end{align*}
			}
			Since the exponential function is locally Lipschitz continuous, $|e^x-e^y|=e^y|e^{x-y}-1|\le2e^y|x-y|$ if $|x-y|\le\ln2$. By the triangle inequality, \ref{assump-2A}, and \ref{assump-2D}, 
			{\small
				\begin{align*}
					\underset{\lob\in\domr}{\sup}\oddsr(\lob;\alpha_{K_r}^*)
					&\le\underset{\lob\in\domr}{\sup}\oddsr(\lob)
					+\underset{\lob\in\domr}{\sup}\left|\oddsr(\lob;\alpha_{K_r}^*)-\oddsr(\lob)\right|\\
					&\le C_0+C_1K_r^{-\mu_1}\ .
				\end{align*}
			}
			Thus, there exists large enough $N^*$ such that $\sup_{\lob\in\domr}\oddsr(\lob;\alpha_{K_r}^*)
			\le 2C_0$ for all $N\ge N^*$. Therefore,
			\begin{align}\label{odds-diff}
				|\exp\left\{\Phir(\lob)\tr\hat{\alpha}^r\right\}-\exp\{\Phir(\lob)\tr\alpha_{K_r}^*\}|
				\le4C_0|\Phir(\lob)\tr\hat{\alpha}^r-\Phir(\lob)\tr\alpha_{K_r}^*|
			\end{align}
			if $|\Phir(\lob)\tr\hat{\alpha}^r-\Phir(\lob)\tr\alpha_{K_r}^*|\le\ln2$. By the Cauchy inequality and \ref{assump-2E}, 
			$|\Phir(\lob)\tr\hat{\alpha}^r-\Phir(\lob)\tr\alpha_{K_r}^*|\le K_r^{1/2}\|\hat{\alpha}^r-\alpha_{K_r}^*\|_2$ for any $\lob\in\domr$. By Lemma \ref{alpha}, 
			$\|\hat{\alpha}^r-\alpha_{K_r}^*\|_2=O_p(\sqrt{K_r/N_r}+K_r^{-\mu_1})$. More precisely, for any $\epsilon>0$, there exists a finite $M_\epsilon>0$ and $N_\epsilon>0$ such that 
			\begin{align*}
				P\left\{\|\hat{\alpha}^r-\alpha_{K_r}^*\|_2
				>M_\epsilon\left(\sqrt{\frac{K_r}{N_r}}+K_r^{-\mu_1}\right)
				\right\}<\epsilon
			\end{align*}
			for any $N\ge N_\epsilon$. Considering the complementary event, we can find $N_\epsilon^*$ large enough such that $M_\epsilon(K_r/\sqrt{N_r}+K_r^{1/2-\mu_1})<\ln2$ which makes the inequality \eqref{odds-diff} hold for any $N\ge N_\epsilon^*$. Then,
			{\footnotesize
				\begin{align*}
					P\left\{\underset{\lob\in\domr}{\sup}\left| \oddsr(\lob;\hat{\alpha}^r)-\oddsr(\lob)\right|
					\le4C_0M_\epsilon\left(\frac{K_r}{\sqrt{N_r}}+K_r^{\frac{1}{2}-\mu_1}\right)+C_1K_r^{-\mu_1}
					\right\}\ge1-\epsilon
				\end{align*}
			}
			for all $N\ge\max\{N^*,N_\epsilon,N_\epsilon^*\}$. In other words, 
			\begin{align*}
				\underset{\lob\in\domr}{\sup}|\oddsr(\lob;\hat{\alpha}^r)-\oddsr(\lob)|=O_p\left(\frac{K_r}{\sqrt{N_r}}+K_r^{\frac{1}{2}-\mu_1}\right)\ .
			\end{align*}
			Now, we consider the $L_2(P)$ norm.
			{\small
				\begin{align*}
					&\quads\|\oddsr(\Lob;\hat{\alpha}^r)-\oddsr(\Lob)\|_{P,2}\\
					&\le\|\oddsr(\Lob;\hat{\alpha}^r)-\oddsr(\Lob;\alpha_{K_r}^*)\|_{P,2}+\|\oddsr(\Lob;\alpha_{K_r}^*)-\oddsr(\Lob)\|_{P,2}\\
					&\le\|\oddsr(\Lob;\hat{\alpha}^r)-\oddsr(\Lob;\alpha_{K_r}^*)\|_{P,2}+\underset{\lob\in\domr}{\sup}\left| \oddsr(\lob;\alpha_{K_r}^*)-\oddsr(\lob)\right|\ .
				\end{align*}
			}
			Following the similar arguments, when $|\Phir(\lob)\tr\hat{\alpha}^r-\Phir(\lob)\tr\alpha_{K_r}^*|\le\ln2$, we have
			\begin{align*}
				&\quads\|\oddsr(\Lob;\hat{\alpha}^r)-\oddsr(\Lob;\alpha_{K_r}^*)\|_{P,2}^2\\
				&\le16C_0^2\int\left\{\Phir(\Lob)\tr\hat{\alpha}^r-\Phir(\Lob)\tr\alpha_{K_r}^*\right\}^2dP(L)\\
				&\le16C_0^2\int(\hat{\alpha}^r-\alpha_{K_r}^*)\tr\Phir(\Lob)\Phir(\Lob)\tr(\hat{\alpha}^r-\alpha_{K_r}^*)dP(L)\\
				&\le16C_0^2\underset{\lob\in\domr}{\sup}\lambda_{\max}\{\Phir(\lob)\Phir(\lob)\tr\}\int(\hat{\alpha}^r-\alpha_{K_r}^*)\tr(\hat{\alpha}^r-\alpha_{K_r}^*)dP(L)\\
				&\le16C_0^2\lambda_{\max}^*\|\hat{\alpha}^r-\alpha_{K_r}^*\|_2^2\ .
			\end{align*}
			Thus, $
			\|\oddsr(\Lob;\hat{\alpha}^r)-\oddsr(\Lob;\alpha_{K_r}^*)\|_{P,2}=O_p(\sqrt{K_r/N_r}+K_r^{-\mu_1})$. Therefore,
			\begin{align*}
				\|\oddsr(\Lob;\hat{\alpha}^r)-\oddsr(\Lob)\|_{P,2}
				=O_p\left(\sqrt{\frac{K_r}{N_r}}+K_r^{-\mu_1}\right)\ .
			\end{align*}
		\end{proof}

		\section{Proof of Theorem \ref{psi}}
		\label{sec:proof-normality}

		{\bf Proof sketch:}
		We further decompose the error terms and show that the first three terms converge to 0 with the rate faster than $1/\sqrt{N}$, and the last term contributes as the influence function. Since the components in the decomposition involve the estimator, they should be treated as random functions. So, we consider the uniform convergence over a set of functions and apply the theory in \citet{van2000asymptotic}. With the maximal inequality via bracketing, the problem is converted to the control of the entropy integral, which requires the calculation of bracketing numbers. 
		Lemmas \ref{bracket1}--\ref{bracket4} are bracketing inequalities which could be of independent interest.

		\begin{proof}
			First, recall that $\hat{\bbP}_N\psi_\theta-\E\{\psi_\theta(L)\}$ has the following decomposition 
			{\small
				\begin{align*}
					\hat{\bbP}_N\psi_\theta-\E\{\psi_\theta(L)\}
					=\sum_{r\in\calR}\left[\frac{1}{N}\sum_{i=1}^N\mathsf{1}_{R_i=1_d}\oddsr(\Lobi;\hat{\alpha}^r)\psi_\theta(L_i)-\E\{\mathsf{1}_{R=r}\psi_\theta(L)\}\right]\ .
				\end{align*}
			}
			For each missing pattern $r$, denote $1/N\sum_{i=1}^N\mathsf{1}_{R_i=1_d}\oddsr(\Lobi;\hat{\alpha}^r)\psi_\theta(L_i)$ as $\hat{\bbP}_N^r\psi_\theta$. Then, $\hat{\bbP}_N^r\psi_\theta-\E\{\mathsf{1}_{R=r}\psi_\theta(L)\}$ can be decomposed into 4 parts:
			\begin{align*}
				S_{\theta,1}^r
				&=\frac{1}{N}\sum_{i=1}^N\mathsf{1}_{R_i=1_d}\left\{\oddsr(\Lobi;\hat{\alpha}^r)-\oddsr(\Lob_i)\right\}\left\{\psi_\theta(L_i)-u^r_\theta(\Lob_i)\right\}\ ,\\
				S_{\theta,2}^r
				&=\frac{1}{N}\sum_{i=1}^N\left\{\mathsf{1}_{R_i=1_d}\oddsr(\Lobi;\hat{\alpha}^r)-\mathsf{1}_{R_i=r}\right\}\left\{u^r_\theta(\Lob_i)-\Phir(\Lob_i)\tr\beta_\theta^r\right\}\ ,\\
				S_{\theta,3}^r
				&=\frac{1}{N}\sum_{i=1}^N\left\{\mathsf{1}_{R_i=1_d}\oddsr(\Lobi;\hat{\alpha}^r)-\mathsf{1}_{R_i=r}\right\}\Phir(\Lob_i)\tr\beta_\theta^r\ ,\\
				S_{\theta,4}^r
				&=\frac{1}{N}\sum_{i=1}^N\mathsf{1}_{R_i=1_d}\oddsr(\Lob_i)\left\{\psi_\theta(L_i)-u^r_\theta(\Lob_i)\right\}\\
				&\quads+\frac{1}{N}\sum_{i=1}^N\mathsf{1}_{R_i=r}u^r_\theta(\Lob_i)-\E\{\mathsf{1}_{R=r}\psi_\theta(L)\}\ .
			\end{align*}
			For any fixed $\theta\in\Theta$ and any missing pattern $r$, by Lemmas \ref{S1}, \ref{S2}, and \ref{S3}, $\sqrt{N}|S_{\theta,i}^r|=o_p(1), i=1,2,3$. It's easy to see that $\E(S_{\theta,4}^r)=0$. Therefore, by the central limit theorem,
			\begin{align*}
				\sqrt{N}\left[\hat{\bbP}_N\psi_\theta-\E\{\psi_\theta(L)\}\right]\to\calN(0,V_\theta)
			\end{align*} 
			where $V_\theta=\E\{F_\theta(L,R)F_\theta(L,R)\tr\}$ and
			{\small
				\begin{align*}
					F_\theta(L,R)=\mathsf{1}_{R=1_d}\sum_{r\in\calR}\oddsr(\Lob)\{\psi_\theta(L)-u_\theta^r(\Lob)\}+\sum_{r\in\calR}\mathsf{1}_{R=r}u_\theta^r(\Lob)-\E\{\psi_\theta(L)\}\ .
				\end{align*}
			}
		\end{proof}

		\section{Proof of Theorem \ref{theta}}

		{\bf Proof sketch:}
		First, by \ref{assump-4A}, the convergence of $\hat{\theta}_N$ should be implied by the uniform convergence of $\hat{\bbP}_N\psi_\theta$ over $\theta$ in a compact set. Second, we study the convergence of $\E\{\psi_{\hat{\theta}_N}(L)\}$ and apply the Delta method to obtain the limiting distribution of $\hat{\theta}_N$. The functional version of the central limit theorem, \emph{i.e.} Donsker's theorem, is applied to achieve uniform convergence.

		\begin{proof}
			Denote the empirical average $N^{-1}\sum_{i=1}^N\psi_\theta(L_i)$ as $\bbP_N\psi_\theta$ and the centered and scaled version $\sqrt{N}[\bbP_N\psi_\theta-\E\{\psi_\theta(L)\}]$ as $\bbG_N\psi_\theta$. Recall the proposed weighted average is
			\begin{align*}
				\hat{\bbP}_N\psi_\theta=\frac{1}{N}\sum_{i=1}^N\left\{\mathsf{1}_{R_i=1_d}\hat{w}(L_i)\psi_\theta(L_i)\right\}\ .
			\end{align*}
			Since $\hat{\theta}_N$ is the solution to $\hat{\bbP}_N\psi_\theta=0$, by Lemma \ref{uniformbound},
			\begin{align*}
				\E\{\psi_{\hat{\theta}_N}(L)\}=\E\{\psi_{\hat{\theta}_N}(L)\}-\hat{\bbP}_N\psi_{\hat{\theta}_N}
				\le\underset{\theta\in\Theta}{\sup}\left|\hat{\bbP}_N\psi_\theta-\E\{\psi_\theta(L)\}\right|=o_p(1)\ .
			\end{align*}
			By identifiability condition \ref{assump-4A}, $\|\hat{\theta}_N-\theta_0\|_2\xrightarrow{P}0$. 
			
			Next, we investigate the asymptotic normality of $\hat{\theta}_N$. Although $\E\{\psi_{\hat{\theta}_N}(L)\}$ has a form of expectation over the population, it can be viewed as a random vector because $\hat{\theta}_N$ depends on the observations. Since $\hat{\theta}_N\xrightarrow{P}\theta_0$, one would expect that $\E\{\psi_{\hat{\theta}_N}(L)\}$ converges to $\E\{\psi_{\theta_0}(L)\}$ in some way. If the limiting distribution is known, one could apply the Delta method to obtain limiting distribution of $\hat{\theta}_N$. From Theorem \ref{psi}, we have the asymptotic normality of $\hat{\bbP}_N\psi_\theta$ for any fixed $\theta\in\Theta$. It is natural to consider
			\begin{align}\label{difference-of-interest}
				\left[\hat{\bbP}_N\psi_{\hat{\theta}_N}-\E\{\psi_{\hat{\theta}_N}(L)\}\right]
				-\left[\hat{\bbP}_N\psi_{\theta_0}-\E\{\psi_{\theta_0}(L)\}\right]
			\end{align}
			The above difference has a similar form to the asymptotic equicontinuity, which can be derived if the function class is Donsker. More precisely, consider the class of $j$-th entry of the estimating functions, $\Psi_j:=\{\psi_{\theta,j}:\theta\in\Theta\}$. It is Donsker due to Theorem 19.5 in \cite{van2000asymptotic} and
			\begin{align*}
				J_{[ \ ]}\{1,\Psi_j,L_2(P)\}
				&=\int_0^1\sqrt{\log n_{[ \ ]}\{\epsilon,\Psi_j,L_2(P)\}}d\epsilon\\
				&\le\int_0^1\sqrt{\log n_{[ \ ]}\{\epsilon,\calH,L_2(P)\}}d\epsilon\\
				&\le\int_0^1\sqrt{C_\calH}\epsilon^{-\frac{1}{2d_\calH}}d\epsilon
				=\sqrt{C_\calH}\le\infty\ .
			\end{align*}
			Then, by Section 2.1.2 in \cite{wellner2013weak}, we have the following asymptotic equicontinuity: for any $\epsilon,\eta>0$, there exists $C_{\epsilon,\eta}>0$ and $N_{\epsilon,\eta}$ such that for all $N\ge N_{\epsilon,\eta}$, 
			\begin{align*}
				P\left(\underset{\psi_{\theta,j}:\rho_P(\psi_{\theta,j}-\psi_{\theta_0,j})<C_{\epsilon,\eta}}{\sup}\left|\bbG_N\psi_{\theta,j}-\bbG_N\psi_{\theta_0,j}\right|\ge\epsilon\right)\le\frac{\eta}{2}
			\end{align*}
			where the seminorm $\rho_P$ is defined as $\rho_P(f)=\{P(f-Pf)^2\}^{1/2}$. Consider
			\begin{align*}
				&\quads\bbG_N\psi_{\hat{\theta}_N,j}-\bbG_N\psi_{\theta_0,j}\\
				&=\sqrt{N}\left[\bbP_N\psi_{\hat{\theta}_N,j}-\E\{\psi_{\hat{\theta}_N,j}(L)\}\right]
				-\sqrt{N}\left[\bbP_N\psi_{\theta_0,j}-\E\{\psi_{\theta_0,j}(L)\}\right]\ .
			\end{align*}
			Notice that $\rho_P(f)\le\|f\|_{P,2}$. By \ref{assump-4B}, for any $\delta>0$, there exists an envelop function $f_{\delta,j}$ such that
			{\small
				\begin{align*}
					P\left(\|\hat{\theta}_N-\theta_0\|_2<\delta\right)
					\le P\left(\|\psi_{\hat{\theta}_N,j}-\psi_{\theta_0,j}\|_{P,2}<C_\delta\right)
					\le P\left\{\rho_P(\psi_{\hat{\theta}_N,j}-\psi_{\theta_0,j})<C_\delta\right\}
				\end{align*}
			}
			where $C_\delta=\|f_{\delta,j}\|_{P,2}\to0$ when $\delta\to0$. Thus, there exists $\delta_{\epsilon,\eta}$ small enough such that $C_\delta\le C_{\epsilon,\eta}$ for all $\delta\le\delta_{\epsilon,\eta}$. Then, by the consistency of $\hat{\theta}_N$, there exists $N_{\epsilon,\eta}^*$ such that for all $N\ge N_{\epsilon,\eta}^*$, 
			\begin{align*}
				P\left(\|\hat{\theta}_N-\theta_0\|_2\ge\delta_{\epsilon,\eta}\right)
				\le\frac{\eta}{2}\ .
			\end{align*}
			Thus, 
			\begin{align*}
				P\left\{\rho_P(\psi_{\hat{\theta}_N,j}-\psi_{\theta_0,j})<C_{\epsilon,\eta}\right\}
				>1-\frac{\eta}{2}\ .
			\end{align*}
			Note that the event $\rho_P(\psi_{\hat{\theta}_N,j}-\psi_{\theta_0,j})<C_{\epsilon,\eta}$ and
			\begin{align*}
				\underset{\psi_{\theta,j}:\rho_P(\psi_{\theta,j}-\psi_{\theta_0,j})<C_{\epsilon,\eta}}{\sup}\left|\bbG_N\psi_{\theta,j}-\bbG_N\psi_{\theta_0,j}\right|<\epsilon
			\end{align*}
			happening together implies $|\bbG_N\psi_{\hat{\theta}_N,j}-\bbG_N\psi_{\theta_0,j}|<\epsilon$. By taking complementary event, for any $N\ge\max\{N_{\epsilon,\eta},N_{\epsilon,\eta}^*\}$, we obtain
			\begin{align*}
				&P\left(\left|\bbG_N\psi_{\hat{\theta}_N,j}-\bbG_N\psi_{\theta_0,j}\right|\ge\epsilon\right)
				\le P\left\{\rho_P(\psi_{\hat{\theta}_N,j}-\psi_{\theta_0,j})\ge C_{\epsilon,\eta}\right\}\\
				&\quad\quad\quad\quad
				+P\left(\underset{\psi_{\theta,j}:\rho_P(\psi_{\theta,j}-\psi_{\theta_0,j})<C_{\epsilon,\eta}}{\sup}\left|\bbG_N\psi_{\theta,j}-\bbG_N\psi_{\theta_0,j}\right|\ge\epsilon\right)\\
				&\le\frac{\eta}{2}+\frac{\eta}{2}=\eta\ .
			\end{align*}
			That is, for each $j$-th entry,
			\begin{align}\label{equicontinuity}
				\sqrt{N}\left[\bbP_N\psi_{\hat{\theta}_N,j}-\E\{\psi_{\hat{\theta}_N,j}(L)\}\right]
				-\sqrt{N}\left[\bbP_N\psi_{\theta_0,j}-\E\{\psi_{\theta_0,j}(L)\}\right]
				\xrightarrow{P}0\ .
			\end{align}
			Therefore, by the comparison between terms in \eqref{difference-of-interest} and $\bbG_N\psi_{\hat{\theta}_N,j}-\bbG_N\psi_{\theta_0,j}$, we should consider
			\begin{align*}
				\sqrt{N}\left[\hat{\bbP}_N\psi_{\hat{\theta}_N,j}-\bbP_N\psi_{\hat{\theta}_N,j}\right]
				-\sqrt{N}\left[\hat{\bbP}_N\psi_{\theta_0,j}-\bbP_N\psi_{\theta_0,j}\right]
			\end{align*}
			which can be decomposed as the following terms. 
			\begin{align*}
				&\quads\sum_{r\in\calR}\left(S_{\hat{\theta}_N,1}^r+S_{\hat{\theta}_N,2}^r+S_{\hat{\theta}_N,3}^r+S_{\hat{\theta}_N,5}^r-S_{\hat{\theta}_N,6}^r\right)\\
				&-\sum_{r\in\calR}\left(S_{\theta_0,1}^r+S_{\theta_0,2}^r+S_{\theta_0,3}^r+S_{\theta_0,5}^r-S_{\theta_0,6}^r\right)
			\end{align*}
			where
			\begin{align*}
				S_{\theta,5}^r&=\frac{1}{N}\sum_{i=1}^N\left\{\mathsf{1}_{R_i=1_d}\oddsr(\Lob_i)-\mathsf{1}_{R_i=r}\right\}\psi_\theta(L_i)\ ,\\
				S_{\theta,6}^r&=\frac{1}{N}\sum_{i=1}^N\left\{\mathsf{1}_{R_i=1_d}\oddsr(\Lob_i)-\mathsf{1}_{R_i=r}\right\}u^r_\theta(\Lob_i)\ .
			\end{align*}
			By Lemmas \ref{S1}, \ref{S2}, and \ref{S3}, $\sqrt{N}\left|S_{\theta,i}^r\right|=o_p(1),i=1,2,3$ for any missing pattern $r$ and $\theta\in\Theta$. Combing with Lemmas \ref{S5} and \ref{S6}, we have
			\begin{align}\label{weighted-equicontinuity}
				\sqrt{N}\left(\hat{\bbP}_N\psi_{\hat{\theta}_N}-\bbP_N\psi_{\hat{\theta}_N}-\hat{\bbP}_N\psi_{\theta_0}+\bbP_N\psi_{\theta_0}\right)
				\xrightarrow{P}0\ .
			\end{align}
			By Equations \eqref{weighted-equicontinuity} and \eqref{equicontinuity}, we have
			\begin{align*}
				\sqrt{N}\left[\hat{\bbP}_N\psi_{\hat{\theta}_N}-\E\{\psi_{\hat{\theta}_N}(L)\}-\hat{\bbP}_N\psi_{\theta_0}+\E\{\psi_{\theta_0}(L)\}\right]
				\xrightarrow{P}0\ .
			\end{align*}
			Since $\hat{\bbP}_N\psi_{\hat{\theta}_N}=0$ and $\E\{\psi_{\theta_0}(L)\}=0$, the above equation can be rewritten as
			\begin{align*}
				\sqrt{N}\left[\E\{\psi_{\hat{\theta}_N}(L)\}-\E\{\psi_{\theta_0}(L)\}+\hat{\bbP}_N\psi_{\theta_0}-\E\{\psi_{\theta_0}(L)\}\right]
				\xrightarrow{P}0\ .
			\end{align*}
			By Theorem \ref{psi},
			\begin{align*}
				\sqrt{N}\left[\hat{\bbP}_N\psi_{\theta_0}-\E\{\psi_{\theta_0}(L)\}\right]
				\overset{d}{\to}N(0,V_{\theta_0})\ .
			\end{align*}
			Since $D_{\theta_0}$ is nonsingular, by multivariate Delta method,
			\begin{align*}
				\sqrt{N}(\hat{\theta}_N-\theta_0)\overset{d}{\to}
				N\left(0,D_{\theta_0}^{-1}V_{\theta_0}D_{\theta_0}^{-1\tr}\right)\ .
			\end{align*}
			Therefore, $\hat{\theta}_N$ is semiparametrically efficient.
			
			Lastly, we look into the estimator for the asymptotic variance. We have the following decomposition:
			\begin{align*}
				\hat{D}_{\hat{\theta}_N}-D_{\theta_0}
				&=\frac{1}{N}\sum_{i=1}^N\left[\mathsf{1}_{R_i=1_d}\left\{\hat{w}(L_i)-\sum_{r\in\calR}\oddsr(\Lob_i)\right\}\dot{\psi}_{\hat{\theta}_N}(L_i)\right]\\
				&\quads+\frac{1}{N}\sum_{i=1}^N\mathsf{1}_{R_i=1_d}\frac{1}{P(R_i=1_d\mid L_i)}\dot{\psi}_{\hat{\theta}_N}(L_i)-D_{\hat{\theta}_N}+D_{\hat{\theta}_N}-D_{\theta_0}\ .
			\end{align*}
			where $\hat{w}(L_i)=\sum_{r\in\calR}\oddsr(\Lobi;\hat{\alpha}^r)$. 
			
			Consider the first term on the right hand side. By Theorem \ref{odds},
			\begin{align*}
				&\quads\|\hat{w}(l)-1/P(R=1_d\mid l)\|_\infty\\
				&\le\sum_{r\in\calR}\|\oddsr(\ \cdot\ ;\hat{\alpha}^r)-\oddsr\|_\infty
				=O_p(\sqrt{K_r^2/N_r}+K_r^{1/2-\mu_1})\ .
			\end{align*}
			Let
			\begin{align*}
				\bfF=\frac{1}{N}\sum_{i=1}^N\mathsf{1}_{R_i=1_d}\dot{\psi}_{\hat{\theta}_N}(L_i)\dot{\psi}_{\hat{\theta}_N}(L_i)\tr\ .
			\end{align*}
			Following the similar arguments in Lemma \ref{linear}, one can see that
			\begin{align*}
				&\quads\left\|\frac{1}{N}\sum_{i=1}^N\left[\mathsf{1}_{R_i=1_d}\left\{\hat{w}(L_i)-\sum_{r\in\calR}\oddsr(\Lob_i)\right\}\dot{\psi}_{\hat{\theta}_N}(L_i)\right]\right\|_2^2\\
				&\le\underset{\theta\in\Theta}{\sup}\frac{1}{N}\sum_{i=1}^N\mathsf{1}_{R_i=1_d}\left\{\hat{w}(L_i)-\sum_{r\in\calR}\oddsr(\Lob_i)\right\}^2\lambda_{\max}\left\{\bfF\right\}\\
				&\le\lambda_{\max}'\|\hat{w}(l)-1/P(R=1_d\mid l)\|_\infty^2=o_p(1)\ .
			\end{align*}
			
			Consider the second term on the right hand side. Notice that $\dot{\psi}_\theta$ is the Jacobian matrix of $\psi_\theta$. We consider all the entries of $\dot{\psi}_\theta$ together and abbreviate the subscripts in the following statements. Define a set of functions $\calF_\Theta:=\{f_{\theta}:\theta\in\Theta\}$ where $f_\theta(L,R):=1_{R=1_d}/P(R=1_d\mid L)\dot{\psi}_\theta(L)$. Similar to Lemma \ref{bracket4}, one can show that
			\begin{align*}
				n_{[ \ ]}\{\epsilon/\delta_0,\calF_\Theta,L_1(P)\}
				\le n_{[ \ ]}\{\epsilon,\calJ_\Theta,L_1(P)\}<\infty
			\end{align*}
			since $1_{R=1_d}/P(R=1_d\mid L)\le1/\delta_0$. Therefore, by Theorem 19.4 in \cite{van2000asymptotic}, $\calF_\Theta$ is P-Glivenko-Cantelli. Since the set $\calF_\Theta$ includes all entries of the Jacobian matrix, we consider the Frobenius/Euclidean norm of a matrix to construct the following convergence result.
			\begin{align*}
				\underset{f_\theta\in\calF_\Theta}{\sup}\left\|\bbP_Nf_\theta-Pf_\theta\right\|_F\xrightarrow{a.s.}0
			\end{align*}
			where $\|\cdot\|_F$ is the Euclidean norm of a matrix. The fact that $\|\cdot\|_2\le\|\cdot\|_F$ implies
			\begin{align*}
				\left\|\frac{1}{N}\sum_{i=1}^N\mathsf{1}_{R_i=1_d}\frac{1}{P(R_i=1_d\mid L_i)}\dot{\psi}_{\hat{\theta}_N}(L_i)-D_{\hat{\theta}_N}\right\|_2=o_p(1)\ .
			\end{align*}
			
			Finally, $D_{\hat{\theta}_N}\xrightarrow{P}D_{\theta_0}$ since $\|\hat{\theta}_N-\theta_0\|_2\xrightarrow{P}0$ and $\dot{\psi}_\theta$ is continuous in a neighborhood of $\theta_0$. Therefore, $\hat{D}_{\hat{\theta}_N}$ is a consistent estimator of $D_{\theta_0}$. 
			
			We skip the details but provide a skeleton of the following proof. Notice that each component of $\hat{F}_i$ converges to the corresponding true value. Therefore, $\hat{F}_i$ and $V_{\hat{\theta}_N}$ are consistent estimators of $F_{\theta_0}(L_i,R_i)$ and $V_{\theta_0}$ respectively. Since $\hat{D}_{\hat{\theta}_N}^{-1}V_{\hat{\theta}_N}\hat{D}_{\hat{\theta}_N}^{-1\tr}$ is a standard sandwich estimator, it is easy to show it is a consistent estimator of the above asymptotic variance.
		\end{proof}
		
		\section{Related Lemmas}
		\label{sec:lemma}
		\begin{lemma}[Weyl's inequality]\label{Weyl}
			Let $\bfA$ and $\bfB$ be $m \times m$ Hermitian matrices and $\bfC=\bfA-\bfB$. Suppose their respective eigenvalues $\mu_i,\nu_i,\rho_i$ are ordered as follows: 
			\begin{align*}
				\bfA:\quad\mu_1\ge\cdots\ge\mu_m\ ,\\
				\bfB:\quad\nu_1\ge\cdots\ge\nu_m\ ,\\
				\bfC:\quad\rho_1\ge\cdots\ge\rho_m\ .
			\end{align*}
			Then, the following inequalities hold.
			\begin{align*}
				\rho_m\le\mu_i-\nu_i\le\rho_1,\quad i=1,\cdots,m\ .
			\end{align*}
			In particular, if $\bfC$ is positive semi-definite, plugging $\rho_m\ge 0$ into the above inequalities leads to
			\begin{align*}
				\mu_i\ge\nu_i,\quad i=1,\cdots,m\ .
			\end{align*}
		\end{lemma}

		\begin{lemma}[Bernstein’s inequality] \label{Bernstein}
			Let $\{\bfA_i\}_{i=1}^N$ be a sequence of independent random matrices with dimensions $d_1\times d_2$. Assume that $\E\{\bfA_i\}=\mathbf{0}_{d_1,d_2}$ and $\|\bfA_i\|_2\le c$ almost surely for all $i=1,\cdots,N$ and some constant $c$. Also assume that
			\begin{align*}
				\max\left\{\left\|\sum_{i=1}^N\E(\bfA_i\bfA_i\tr)\right\|_2,\left\|\sum_{i=1}^N\E(\bfA_i\tr \bfA_i)\right\|_2\right\}\le\sigma^2\ .
			\end{align*}
			Then, for all $t \ge 0$,
			\begin{align*}
				P\left(\left\|\sum_{i=1}^N\bfA_i\right\|_2\ge t\right)\le(d_1+d_2)\exp\left(-\frac{t^2/2}{\sigma^2+ct/3}\right)\ .
			\end{align*}
		\end{lemma}
		
		\begin{lemma}\label{alpha}
			Under Assumptions \ref{assump1} and \ref{assump2}, the minimizer $\hat{\alpha}^r$ satisfies 
			\begin{align*}
				\|\hat{\alpha}^r-\alpha_{K_r}^*\|_2=O_p\left(\sqrt{\frac{K_r}{N_r}}+K_r^{-\mu_1}\right)=o_p(1)\ .
			\end{align*}
		\end{lemma}
		
		\begin{proof}
			It suffices to show for any $\epsilon>0$, there exists $C_\epsilon$ and $N_\epsilon$ such that 
			\begin{align}\label{Op-alpha}
				P\left\{\|\hat{\alpha}^r-\alpha_{K_r}^*\|_2>C_\epsilon\left(\sqrt{\frac{K_r}{N_r}}+K_r^{-\mu_1}\right)\right\}\le\epsilon
			\end{align}
			for any $N\ge N_\epsilon$. It means that the minimizer $\hat{\alpha}^r$ is in a small neighbourhood of $\alpha_{K_r}^*$ with probability higher than $1-\epsilon$. Consider the set $\Delta=\{\delta\in\bbR^{K_r}:\|\delta\|_2\le C(\sqrt{K_r/N_r}+K_r^{-\mu_1})\}$ for an arbitrary constant $C$. Since $\calLr_\lambda$ is a convex function of $\alphar$, the minimizer $\hat{\alpha}^r\in\alpha_{K_r}^*+\Delta$ if $\inf_{\delta\in\partial\Delta}{\inf}\calLr_\lambda(\alpha_{K_r}^*+\delta)> \calLr_\lambda(\alpha_{K_r}^*)$. Thus, considering the complementary event, we have
			{\small
				\begin{align*}
					P\left\{\|\hat{\alpha}^r-\alpha_{K_r}^*\|_2>C\left(\sqrt{\frac{K_r}{N_r}}+K_r^{-\mu_1}\right)\right\}
					\le
					P\left\{\underset{\delta\in\partial\Delta}{\inf}\calLr_\lambda(\alpha_{K_r}^*+\delta)\le\calLr_\lambda(\alpha_{K_r}^*)\right\}\ .
				\end{align*} 
			}
			Recall that for any $r\neq 1_d$ and any $\lambda>0$, the objective function is
			\begin{align*}
				\calLr_\lambda(\alphar)
				&=\frac{1}{N}\sum_{i=1}^N\left\{\mathsf{1}_{R_i=1_d}\oddsr(\Lobi;\alphar)-\mathsf{1}_{R_i=r}\log\oddsr(\Lobi;\alphar)\right\}+\lambda J^r(\alphar)
			\end{align*}
			where $\oddsr(\lob;\alphar)=\exp\{\Phir(\lob)\tr\alphar\}$ and $J^r(\alphar)=\gamma \sum_{k=1}^{K_r}t_k|\alphar_k|+(1-\gamma)(\alphar)\tr\bfDr\alphar$. To investigate $\inf_{\delta\in\partial\Delta}\calLr_\lambda(\alpha_{K_r}^*+\delta)-\calLr_\lambda(\alpha_{K_r}^*)$, we apply the mean value theorem. There exists some $\tilde{\alpha}^r$ satisfying $\tilde{\alpha}^r-\alpha_{K_r}^*\in\mathrm{int}(\Delta)$, which is the interior of $\Delta$, such that for any $\delta\in\Delta$, 
			\begin{align*}
				&\quads\calLr_\lambda(\alpha_{K_r}^*+\delta)-\calLr_\lambda(\alpha_{K_r}^*)\\
				&=\delta\tr\left.\frac{\partial\calLr_N(\alphar)}{\partial\alphar}\right\vert_{\alpha_{K_r}^*}
				+\frac{1}{2}\delta\tr\left\{\left.\frac{\partial^2\calLr_N(\alphar)}{(\partial\alphar)^2}\right\vert_{\tilde{\alpha}^r}\right\}\delta
				+\lambda J^r(\alpha_{K_r}^*+\delta)-\lambda J^r(\alpha_{K_r}^*)\ .
			\end{align*}
			By the triangle inequality and Cauchy inequality, the difference between penalties satisfies
			\begin{align*}
				&J^r(\alpha_{K_r}^*+\delta)-J^r(\alpha_{K_r}^*)\\
				&=\gamma\sum_{k=1}^{K_r}\left(\left|\alphar_{*,k}+\delta_k^r\right|-\left|\alphar_{*,k}\right|\right)t_k
				+(1-\gamma)\left(\delta\tr\bfDr\delta+2\delta\tr\bfDr\alpha_{K_r}^*\right)\\
				&\ge-\gamma\sum_{k=1}^{K_r}\left|\delta_k^r\right|t_k
				+(1-\gamma)\left(\delta\tr\bfDr\delta+2\delta\tr\bfDr\alpha_{K_r}^*\right)\\
				&\ge-\gamma\sqrt{\sum_{k=1}^{K_r}t_k^2}\left\|\delta\right\|_2
				-2(1-\gamma)\|\delta\|_2\|\bfDr\|_2\|\alpha_{K_r}^*\|_2
				+(1-\gamma)\delta\tr\bfDr\delta\ .
			\end{align*}
			Denote the constant $\gamma\sqrt{\sum_{k=1}^{K_r}t_k^2}+2(1-\gamma)\|\bfDr\|_2\|\alpha_{K_r}^*\|_2$ as $c_{\mathrm{lin}}$. Then, by the Cauchy inequality again, 
			\begin{align*}
				\calLr_\lambda(\alpha_{K_r}^*+\delta)-\calLr_\lambda(\alpha_{K_r}^*)
				&\ge
				-\left\{\left\|\left.\frac{\partial\calLr_N(\alphar)}{\partial\alphar}\right\vert_{\alpha_{K_r}^*}\right\|_2+\lambda c_{\mathrm{lin}}\right\}
				\left\|\delta\right\|_2\\
				&\quads+\delta\tr
				\left\{\frac{1}{2}\left.\frac{\partial^2\calLr_N(\alphar)}{(\partial\alphar)^2}\right\vert_{\tilde{\alpha}^r}+\lambda(1-\gamma)\bfDr\right\}
				\delta\ .
			\end{align*}
			First, let's have a look at the quadratic coefficient. Since $\bfDr$ is a positive semi-definite matrix, $\lambda(1-\gamma)\delta\tr\bfDr\delta\ge0$.
			Thus, by Lemma \ref{quadratic}, the quadratic terms are bounded from below. More precisely, for any $\epsilon>0$, there exists $N_\epsilon^*$ such that for any $N\ge N_\epsilon^*$,
			\begin{align*}
				P\left[\delta\tr\left\{\frac{1}{2}\left.\frac{\partial^2\calLr_N(\alphar)}{(\partial\alphar)^2}\right\vert_{\tilde{\alpha}^r}+\lambda(1-\gamma)\bfDr\right\}\delta
				\ge\frac{C_{\mathrm{quad}}}{2}\|\delta\|_2^2\right]\ge1-\frac{1}{2}\epsilon\ .
			\end{align*} 
			Next, let's investigate the bound of the linear coefficient. By \ref{assump-2E}, $\lambda=O(\sqrt{K_r/N})$. By Lemma \ref{linear}, for any $\epsilon>0$, there exists $N_\epsilon'$ and a constant $C_\epsilon'$ such that for any $N\ge N_\epsilon'$,
			\begin{align*}
				P\left[
				\left\{\left\|\left.\frac{\partial\calLr_N(\alphar)}{\partial\alphar}\right\vert_{\alpha_{K_r}^*}\right\|_2+\lambda c_{\mathrm{lin}}\right\}
				\ge C_\epsilon'\left(\sqrt{\frac{K_r}{N}}+K_r^{-\mu_1}\right)\right]
				\le\frac{1}{2}\epsilon\ .
			\end{align*} 
			Considering the complement of the above event and the fact that $P(A\cap B)=P(A)+P(B)-P(A\cup B)\ge P(A)+P(B)-1$, we have
			{\small
				\begin{align*}
					P\left\{\calLr_\lambda(\alpha_{K_r}^*+\delta)-\calLr_\lambda(\alpha_{K_r}^*)
					\ge\frac{C_{\mathrm{quad}}}{2}\|\delta\|_2^2
					-C_\epsilon'\left(\sqrt{\frac{K_r}{N}}+K_r^{-\mu_1}\right)\|\delta\|_2
					\right\}\ge 1-\epsilon
				\end{align*}
			}
			for any $N\ge\max\{N_\epsilon^*,N_\epsilon'\}$. Note that $\partial\Delta=\{\delta\in\bbR^{K_r}:\|\delta\|_2=C(\sqrt{K_r/N_r}+K_r^{-\mu_1})\}$. Choosing $C>2C_\epsilon'/C_{\mathrm{quad}}$, we have $P\{\inf_{\delta\in\partial\Delta}\calLr_\lambda(\alpha_{K_r}^*+\delta)\ge\calLr_\lambda(\alpha_{K_r}^*)\}\ge1-\epsilon$ for any $\epsilon>0$. Therefore, inequality \eqref{Op-alpha} holds which completes the proof.
		\end{proof}
		
		\begin{lemma}\label{quadratic}
			There exists a constant $C_{\mathrm{quad}}$ such that the Hessian matrix of $\calLr_N(\alphar)$ at $\tilde{\alpha}^r$ satisfies
			\begin{align*}
				\lim_{N\to\infty}P\left[
				\lambda_{\min}\left\{\left.\frac{\partial^2\calLr_N(\alphar)}{(\partial\alphar)^2}\right\vert_{\tilde{\alpha}^r}\right\}
				\ge C_{\mathrm{quad}}
				\right]=1
			\end{align*} 
			where $\lambda_{\min}(\cdot)$ represents the minimal eigenvalue of the matrix. 
		\end{lemma}
		
		\begin{proof}
			Denote the Hessian matrix of $\calLr_N(\alphar)$ at $\tilde{\alpha}^r$ as $\bfA$:
			\begin{align*}
				\bfA=\left.\frac{\partial^2\calLr_N(\alphar)}{(\partial\alphar)^2}\right\vert_{\tilde{\alpha}^r}
				=\frac{1}{N}\sum_{i=1}^N\left\{\oddsr(\Lobi;\tilde{\alpha}^r)\mathsf{1}_{R_i=1_d}\Phir(\Lob_i)\Phir(\Lob_i)\tr\right\}\ .
			\end{align*} 
			Recall that the set $\Delta=\{\delta\in\bbR^{K_r}:\|\delta\|_2\le C(\sqrt{K_r/N_r}+K_r^{-\mu_1})\}$ and $\tilde{\alpha}^r-\alpha_{K_r}^*\in\mathrm{int}(\Delta)$. Following the similar arguments in the proof of Theorem \ref{odds} (Appendix \ref{sec:proof-odds}), it can be easily shown that
			\begin{align*}
				&\quads\oddsr(\lob;\tilde{\alpha}^r)\\
				&\ge\oddsr(\lob)-\left|\oddsr(\lob)-\oddsr(\lob;\alpha_{K_r}^*)\right|
				-4C_0\left|\Phir(\lob)\tr\tilde{\alpha}^r-\Phir(\lob)\tr\alpha_{K_r}^*\right|\\
				&\ge c_0-C_1K_r^{-\mu_1}-4C_0C(K_r/\sqrt{N_r}+K_r^{\frac{1}{2}-\mu_1})\ .
			\end{align*}
			Then, there exists $N_\Delta$ such that $\oddsr(\lob;\tilde{\alpha}^r)>c_0/2$ holds for any $N\ge N_\Delta$. Let 
			\begin{align*}
				\bfB&=\frac{1}{N}\sum_{i=1}^N\left\{\frac{1}{2}c_0\mathsf{1}_{R_i=1_d}\Phir(\Lob_i)\Phir(\Lob_i)\tr\right\},\\
				\bfC&=\frac{1}{2}c_0\E\left\{\mathsf{1}_{R=1_d}\Phir(\Lob)\Phir(\Lob)\tr\right\}
				=\frac{1}{2}c_0\E\left\{P(R=1_d\mid\Lob)\Phir(\Lob)\Phir(\Lob)\tr\right\},\\
				\bfD&=\frac{1}{2}c_0\delta_0\E\left\{\Phir(\Lob)\Phir(\Lob)\tr\right\}\ .
			\end{align*} 
			It's easy to see that matrices $\bfA,\bfB,\bfC,\bfD$ are symmetric. Based on the above discussions, $\bfA-\bfB$ is positive semi-definite for large enough $N$. By \ref{assump-1B}, $\bfC-\bfD$ is also positive semi-definite. Applying Lemma \ref{Weyl}, we have $\lambda_{\min}(\bfA)\ge\lambda_{\min}(\bfB)$, $\lambda_{\min}(\bfC)\ge\lambda_{\min}(\bfD)$ and $|\lambda_{\min}(\bfB)-\lambda_{\min}(\bfC)|\le\max\{|\lambda_{\min}(\bfB-\bfC)|,|\lambda_{\max}(\bfB-\bfC)|\}=\|\bfB-\bfC\|_2$.
			Therefore, $\lambda_{\min}(\bfA)\ge\lambda_{\min}(\bfD)-\|\bfB-\bfC\|_2\ge c_0\delta_0\lambda_{\min}^*/2-\|\bfB-\bfC\|_2$. To study $\|\bfB-\bfC\|_2$, we apply Lemma \ref{Bernstein}. Let 
			\begin{align*}
				\bfE_i=\frac{1}{N}\Big[\mathsf{1}_{R_i=1_d}\Phir(\Lob_i)\Phir(\Lob_i)\tr-\E\left\{\mathsf{1}_{R=1_d}\Phir(\Lob)\Phir(\Lob)\tr\right\}\Big]\ .
			\end{align*} 
			So, $\E\{\bfE_i\}=\mathbf{0}_{K_r,K_r}$. By the triangle inequality, Lemma \ref{Weyl} and the fact that $\|\cdot\|_2\le\|\cdot\|_F$,
			{\small
				\begin{align*}
					\|\bfE_i\|_2
					&\le\frac{1}{N}\mathsf{1}_{R_i=1_d}\|\Phir(\Lob_i)\Phir(\Lob_i)\tr\|_F
					+\frac{1}{N}\|\E\left\{\mathsf{1}_{R=1_d}\Phir(\Lob)\Phir(\Lob)\tr\right\}\|_2\\
					&\le\frac{1}{N}\sqrt{\textrm{trace}\{\Phir(\Lob_i)\Phir(\Lob_i)\tr\Phir(\Lob_i)\Phir(\Lob_i)\tr\}}
					+\frac{1}{N}\|\E\left\{\Phir(\Lob)\Phir(\Lob)\tr\right\}\|_2\\
					&=\frac{1}{N}\|\Phir(\Lob_i)\|_2^2
					+\frac{1}{N}\|\E\left\{\Phir(\Lob)\Phir(\Lob)\tr\right\}\|_2\ .
				\end{align*} 
			}
			By \ref{assump-2E}, \ref{assump-2F} and the fact that $N_r/N<1$, $\|\bfE_i\|_2=O(K_r/N_r)$. Similarly,
			\begin{align*}
				&\quads\left\|\sum_{i=1}^N\E(\bfE_i\bfE_i\tr)\right\|_2\\
				&\le\frac{1}{N}\left\|\E\{\mathsf{1}_{R=1_d}\Phir(\Lob)\Phir(\Lob)\tr\Phir(\Lob)\Phir(\Lob)\tr\}\right\|_2\\
				&\quads+\frac{1}{N}\left\|\E\{\mathsf{1}_{R=1_d}\Phir(\Lob)\Phir(\Lob)\tr\}\E\left\{\mathsf{1}_{R=1_d}\Phir(\Lob)\Phir(\Lob)\tr\right\}\right\|_2\\
				&\le\frac{1}{N}\underset{\lob\domr}{\sup}\|\Phir(\lob)\|_2^2\|\E\left\{\Phir(\Lob)\Phir(\Lob)\tr\right\}\|_2
				+\frac{1}{N}\|\E\left\{\Phir(\Lob)\Phir(\Lob)\tr\right\}\|_2^2\\
				&=O(K_r/N_r)\ .
			\end{align*} 
			Taking $t=C\sqrt{K_r\log K_r/N_r}$ in Lemma \ref{Bernstein} for an arbitrary constant $C$, we have
			\begin{align*}
				\exp\left(\left\|\sum_{i=1}^N\bfE_i\right\|_2\ge t\right)\le2K_r\exp(-C'\log K_r)
			\end{align*}
			for large enough $N$ and some constant $C'$. In other words, $\|\bfB-\bfC\|_2=O_p(\sqrt{K_r\log K_r/N_r})=o_p(1)$. Therefore, for any $\epsilon$ there exists $N_{\Delta,\epsilon}$ such that 
			\begin{align*}
				P\left\{\lambda_{\min}(\bfA)\ge\frac{1}{4}c_0\delta_0\lambda_{\min}^*\right\}
				\ge1-\epsilon
			\end{align*}
			for any $N\ge\max\{N_\Delta,N_{\Delta,\epsilon}\}$.
		\end{proof}

		\begin{lemma}\label{linear}
			The gradient of $\calLr_N(\alphar)$ at $\alpha_{K_r}^*$ satisfies
			\begin{align*}
				\left\|\left.\frac{\partial\calLr_N(\alphar)}{\partial\alphar}\right\vert_{\alpha_{K_r}^*}\right\|_2
				&=O_p\left(\sqrt{\frac{K_r}{N}}+K_r^{-\mu_1}\right)\ .
			\end{align*} 
		\end{lemma}
		
		\begin{proof}
			The gradient of $\calLr_N(\alphar)$ at $\alpha_{K_r}^*$ is
			\begin{align*}
				\left.\frac{\partial\calLr_N(\alphar)}{\partial\alphar}\right\vert_{\alpha_{K_r}^*}
				=\frac{1}{N}\sum_{i=1}^N\left\{\mathsf{1}_{R_i=r}-\mathsf{1}_{R_i=1_d}\oddsr(\Lobi;\alpha_{K_r}^*)\right\}\Phir(\Lob_i)\ .
			\end{align*}
			Thus, by the triangle inequality,
			{\small
				\begin{align*}
					\left\|\left.\frac{\partial\calLr_N(\alphar)}{\partial\alphar}\right\vert_{\alpha_{K_r}^*}\right\|_2
					&\le\left\|\frac{1}{N}\sum_{i=1}^N\left\{\mathsf{1}_{R_i=r}-\mathsf{1}_{R_i=1_d}\oddsr(\Lob_i)\right\}\Phir(\Lob_i)\right\|_2\\
					&\quads+\left\|\frac{1}{N}\sum_{i=1}^N\mathsf{1}_{R_i=1_d}\left\{\oddsr(\Lob_i)-\oddsr(\Lobi;\alpha_{K_r}^*)\right\}\Phir(\Lob_i)\right\|_2\ .
				\end{align*} 
			}
			Consider the first term on the right hand side. Let $A_i=\{\mathsf{1}_{R_i=r}-\mathsf{1}_{R_i=1_d}\oddsr(\Lob_i)\}\Phir(\Lob_i)$. It's easy to see that $\{A_i\}_{i=1}^N$ are i.i.d. and $\E(A_i)=0$. Thus, 
			\begin{align*}
				&\quads\E\left\|\frac{1}{N}\sum_{i=1}^NA_i\right\|_2^2
				=\frac{1}{N}\E(A_i\tr A_i)\\
				&=\frac{1}{N}\E\left[\sum_{k=1}^{K_r}
				\left\{\mathsf{1}_{R=r}+\mathsf{1}_{R=1_d}\oddsr(\Lob)\oddsr(\Lob)\right\}
				\phir_k(\Lob)\phir_k(\Lob)\right]\\
				&\le\frac{C_0^2+1}{N}\E\|\Phir(\Lob)\|_2^2\ .
			\end{align*}
			By \ref{assump-2E}, $\E\|\sum_{i=1}^NA_i/N\|_2^2=O(K_r/N)$. By the Markov inequality, this implies $\|\sum_{i=1}^NA_i/N\|_2=O_p(\sqrt{K_r/N})$. As for the second term on the right hand side, let $\xi=(\xi_1,\cdots,\xi_N)$ where $\xi_i=\mathsf{1}_{R_i=1_d}\{\oddsr(\Lob_i)-\oddsr(\Lobi;\alpha_{K_r}^*)\}$ and $B=(B_1,\cdots,B_N)$ where $B_i=\mathsf{1}_{R_i=1_d}\Phir(\Lob_i)$. Then,
			\begin{align*}
				\left\|\frac{1}{N}\sum_{i=1}^N\xi_iB_i\right\|_2^2
				=\frac{1}{N^2}\xi\tr BB\tr\xi
				=\frac{1}{N}\xi\tr\left\{\frac{1}{N}\sum_{i=1}^N\mathsf{1}_{R_i=1_d}\Phir(\Lob_i)\Phir(\Lob_i)\tr\right\}\xi\ .
			\end{align*}
			Following the similar arguments in the proof of Lemma \ref{quadratic}, it's easy to see that 
			\begin{align*}
				\lambda_{\max}\left\{\frac{1}{N}\sum_{i=1}^N\mathsf{1}_{R_i=1_d}\Phir(\Lob_i)\Phir(\Lob_i)\tr\right\}
				\le\lambda_{\max}^*+o_p(1)\ .
			\end{align*}
			By \ref{assump-2D}, $|\xi_i|\le C_1K_r^{-\mu_1}$. Thus, $\|\frac{1}{N}\sum_{i=1}^N\xi_iB_i\|_2=O_p(K_r^{-\mu_1})$ and 
			\begin{align*}
				\left\|\left.\frac{\partial\calLr_N(\alphar)}{\partial\alphar}\right\vert_{\alpha_{K_r}^*}\right\|_2
				=O_p\left(\sqrt{\frac{K_r}{N}}+K_r^{-\mu_1}\right)\ .
			\end{align*}
		\end{proof}

		\begin{lemma}\label{S1}
			Under Assumptions \ref{assump1}--\ref{assump3}, for any missing pattern $r$, 
			$$
			\sup_{\theta\in\Theta}|\sqrt{N}S_{\theta,1}^r|=o_p(1).
			$$
		\end{lemma}
		
		\begin{proof}
			Consider the following empirical process.
			\begin{align*}
				\bbG_N(f_{\theta,1})=\sqrt{N}\left[\frac{1}{N}\sum_{i=1}^Nf_{\theta,1}(L_i,R_i)-\E\left\{f_{\theta,1}(L,R)\right\}\right]
			\end{align*}
			where $f_{\theta,1}(L,R)=\mathsf{1}_{R=1_d}\{O(\Lob)-\oddsr(\Lob)\}\{\psi_\theta(L)-u^r_\theta(\Lob)\}$ and $O$ is an arbitrary function, which can be viewed as an estimator of true propensity odds $\oddsr$. By Theorem \ref{odds}, for any $\gamma>0$, there exists constants $C_\gamma>0$ and $N_\gamma >0$ such that for any $N\ge N_\gamma$,
			\begin{align*}
				P\left\{\left\|\oddsr(\ \cdot\ ;\hat{\alpha}^r)-\oddsr\right\|_\infty\ge C_\gamma\left(\sqrt{\frac{K_r^2}{N_r}}+K_r^{\frac{1}{2}-\mu_1}\right)\right\}
				\le\gamma.
			\end{align*}
			Let $\delta_1=C_\gamma(\sqrt{K_r^2/N_r}+K_r^{1/2-\mu_1})$ and consider the set of functions 
			\begin{align*}
				\calF_1=\left\{f_{\theta,1}:\left\|O-\oddsr\right\|_\infty\le\delta_1,\theta\in\Theta\right\}\ .
			\end{align*}
			By \ref{assump-1C}, for any $f_{\theta,1}\in\calF_1$, 
			\begin{align*}
				\E\left\{f_{\theta,1}(L,R)\right\}
				&=\E\left[\E\left\{f_{\theta,1}(L,R)\mid L^r,R\right\}\right]\\
				&=\E\left[\mathsf{1}_{R=1_d}\E\left\{f_{\theta,1}(L,R)\mid L^r,R=1_d\right\}\right]
				=0\ .
			\end{align*}
			Define $\hat{f}_{\theta,1}(L,R):=\mathsf{1}_{R=1_d}\{\oddsr(\Lob;\hat{\alpha}^r)-\oddsr(\Lob)\}\{\psi_\theta(L)-u^r_\theta(\Lob)\}$. To simplify notations, vectors $A>B$ means that $A_j>B_j$ for each entry, and vector $A>c$ means that $A_j>c$ for each entry where $c$ is a constant. 
			
			Notice that $\sup_{\theta\in\Theta}|\sqrt{N}S_{\theta,1}^r|=\sup_{\theta\in\Theta}|\bbG_N(\hat{f}_{\theta,1})|$. Thus, 
			\begin{align*}
				1-\gamma
				\le P\left(\hat{f}_{\theta,1}\in\calF_1\right)
				\le P\left(\underset{\theta\in\Theta}{\sup}\left|\sqrt{N}S_{\theta,1}^r\right|\le\underset{f_{\theta,1}\in\calF_1}{\sup}|\bbG_N(f_{\theta,1})|\right)\ .
			\end{align*}
			By Markov's inequality, for any $\xi>0$, we have 
			\begin{align*}
				P\left(\underset{f_{\theta,1}\in\calF_1}{\sup}|\bbG_N(f_{\theta,1})|
				\ge\frac{1}{\xi}\E\underset{f_{\theta,1}\in\calF_1}{\sup}\left|\bbG_N(f_{\theta,1})\right|\right)
				\le\xi\ .
			\end{align*}
			If we can show $\E\sup_{f_{\theta,1}\in\calF_1}|\bbG_N(f_{\theta,1})|=o_p(1)$, then for any $\eta>0$ and fixed $\xi>0$, there exists $N_{\xi,\eta}$ and $\sigma_{\xi,\eta}$ such that for any $N\ge N_{\xi,\eta}$,
			\begin{align*}
				P\left(\frac{1}{\xi}\E\underset{f_{\theta,1}\in\calF_1}{\sup}\left|\bbG_N(f_{\theta,1})\right|\ge\sigma_{\xi,\eta}\right)\le\eta\ .
			\end{align*}
			Then, for any $\epsilon>0$, by taking $\gamma=\xi=\eta=\frac{\epsilon}{3}$ and appropriately choosing $C_\gamma$, $N_\gamma$, $N_{\xi,\eta}$ and $\sigma_{\xi,\eta}$, we have the above inequalities and for any $N\ge N_\epsilon=\max\{N_\gamma,N_{\xi,\eta}\}$,
			\begin{align*}
				P\left(\underset{\theta\in\Theta}{\sup}\left|\sqrt{N}S_{\theta,1}^r\right|\ge\sigma_{\xi,\eta}\right)
				\le\gamma+\xi+\eta=\epsilon\ .
			\end{align*}
			That is, $\sup_{\theta\in\Theta}|\sqrt{N}S_{\theta,1}^r|=o_p(1)$.
			
			To show $\E\sup_{f_{\theta,1}\in\calF_1}|\bbG_N(f_{\theta,1})|=o_p(1)$, we utilize the maximal inequality with bracketing (Corollary 19.35 in \cite{van2000asymptotic}). Define the envelop function $F_1(L):=\sup_{\theta\in\Theta}|\psi_\theta(L)-u^r_\theta(\Lob)|\delta_1$. It's easy to see $|f_{\theta,1}(L,R)|\le F_1(L)$ for any $f_{\theta,1}\in\calF_1$. Besides, due to \ref{assump-3D}, for each $j$-th entry,
			\begin{align*}
				\|F_{1,j}\|_{P,2}
				=\sqrt{\int F_{1,j}(L)^2dP(L)}
				=\sqrt{\E\left[\underset{\theta}{\sup}\left\{\psi_{\theta,j}(L)-u_{\theta,j}^r(\Lob)\right\}^2\delta_1^2\right]}
				\le C_3\delta_1\ .
			\end{align*}
			To save notations, $\psi_\theta$ and $u_\theta^r$ are used as their $j$-th entry. We also omit the subscripts ``$j$'' of some sets of functions where the related inequalities should hold for each $j$-th entry.
			
			By the maximal inequality,
			\begin{align*}
				\E\underset{f_{\theta,1}\in\calF_1}{\sup}\left|\bbG_N(f_{\theta,1})\right|
				=O_p\left(J_{[ \ ]}\{C_3\delta_1,\calF_1,L_2(P)\}\right)\ .
			\end{align*}
			To study the entropy integral of $\calF_1$, we split function $f_{\theta,1}$ into two parts and consider two sets of functions $\calG_1=\{g_1:\|g_1\|_\infty\le\delta_1\}$ where $g_1(L)=O(\Lob)-\oddsr(\Lob)$ and $\calH_1=\{h_{\theta,1}:\theta\in\Theta\}$ where $h_{\theta,1}(L)=\psi_\theta(L)-u^r_\theta(\Lob)$. Notice that $\|g_1\|_\infty\le\delta_1$, $\|h_{\theta,1}(L)\|_{P,2}\le C_3$ and $\delta_1\le1$ when $N$ is large enough. By Lemma \ref{bracket1},
			\begin{align*}
				n_{[ \ ]}\left\{4\left(C_3+1\right)\epsilon,\calF_1,L_2(P)\right\}
				\le n_{[ \ ]}\{\epsilon,\calG_1,L^\infty\}n_{[ \ ]}\{\epsilon,\calH_1,L_2(P)\}\ .
			\end{align*}
			Define $\tilde{\calG}_1:=\{g_1:\|g_1\|_\infty\le C\}$ for some constant $C$ and $\calO:=\tilde{\calG}_1+\oddsr=\{O:\|O-\oddsr\|_\infty\le C\}$. It is obvious that $\calG=\delta_1/C\tilde{\calG}$. Since $\oddsr$ is a fixed function,
			\begin{align*}
				&\quads n_{[ \ ]}\left\{\epsilon,\calG_1,L^\infty\right\}
				=n_{[ \ ]}\left\{\epsilon,\delta_1/C\tilde{\calG}_1,L^\infty\right\}\\
				&=n_{[ \ ]}\left\{C\epsilon/\delta_1,\tilde{\calG}_1,L^\infty\right\}
				=n_{[ \ ]}\left\{C\epsilon/\delta_1,\calO,L^\infty\right\}\ .
			\end{align*}
			The true propensity score odds $\oddsr$ is unknown, but its roughness is controlled by \ref{assump-3B}. Thus, we should not consider much more rough functions. In other words, our models for propensity score odds should satisfy a similar smoothness condition. There exists appropriate constant $C_\calO$ such that $\calO\subset\calM^r$. Thus,
			\begin{align*}
				n_{[ \ ]}\{\epsilon,\calO,L^\infty\}
				\le n_{[ \ ]}\{\epsilon,\calM^r,L^\infty\}\ .
			\end{align*}
			Define a set of functions $\calU^r=\{u^r_\theta:\theta\in\Theta\}$. Notice that $\calH_1\subset\Psi-\calU^r$ By Lemma \ref{bracket2}, Assumption \ref{assump-3E} and Lemma \ref{bracket3},
			\begin{align*}
				n_{[ \ ]}\{2\epsilon,\calH_1,L_2(P)\}
				\le n_{[ \ ]}\{\epsilon,\calH,L_2(P)\}n_{[ \ ]}\{\epsilon,\calU^r,L_2(P)\}
				\le n_{[ \ ]}\{\epsilon,\calH,L_2(P)\}^2\ .
			\end{align*}
			Combine the above inequalities and recall \ref{assump-3B} and \ref{assump-3C},
			\begin{align*}
				J_{[ \ ]}\{\|F_1\|_{P,1},\calF_1,L_2(P)\}
				&\le\int_0^{C_3\delta_1}\sqrt{\log n_{[ \ ]}\left\{\frac{C_\calO\epsilon}{4(C_3+1)\delta_1},\calM^r,L^\infty\right\}}d\epsilon\\
				&\quads+\sqrt{2}\int_0^{C_3\delta_1}\sqrt{\log n_{[ \ ]}\left\{\frac{\epsilon}{8(C_3+1)\delta_1},\calH,L_2(P)\right\}}d\epsilon\\
				&\le\sqrt{C_\calM}\int_0^{C_3\delta_1}\{4(C_3+1)\delta_1/(C_\calO\epsilon)\}^{\frac{1}{2d_\calM}}d\epsilon\\
				&\quads+\sqrt{2C_\calH}\int_0^{C_3\delta_1}\{8(C_3+1)\delta_1/\epsilon\}^{\frac{1}{2d_\calH}}d\epsilon\\
				&=\sqrt{C_\calM}\{4(C_3+1)/C_\calO\}^{\frac{1}{2d_\calM}}C_3^{1-\frac{1}{2d_\calM}}\delta_1\\
				&\quads+\sqrt{2C_\calH}\{8(C_3+1)\}^{\frac{1}{2d_\calH}}C_3^{1-\frac{1}{2d_\calH}}\delta_1\\
				&\to0
			\end{align*}
			since $d_\calM,d_\calH>1/2$ and $\delta_1\to0$ as $N\to\infty$. Therefore, $\E\sup_{f_{\theta,1}\in\calF_1}|\bbG_N(f_{\theta,1})|=O_p(o_p(1))=o_p(1)$ and $\sup_{\theta\in\Theta}|\sqrt{N}S_{\theta,1}^r|=o_p(1)$.
		\end{proof}

		\begin{lemma}\label{S2}
			Under Assumptions \ref{assump1}--\ref{assump3}, for any missing pattern $r$, $\sup_{\theta\in\Theta}|\sqrt{N}S_{\theta,2}^r|=o_p(1)$.
		\end{lemma}
		
		\begin{proof}
			Consider the following empirical process.
			\begin{align*}
				\bbG_N(f_{\theta,2})=\sqrt{N}\left[\frac{1}{N}\sum_{i=1}^Nf_{\theta,2}(L_i,R_i)-\E\left\{f_{\theta,2}(L,R)\right\}\right]
			\end{align*}
			where $f_{\theta,2}(L,R)=\{\mathsf{1}_{R=1_d}O(\Lob)-\mathsf{1}_{R=r}\}
			\{u^r_\theta(\Lob)-U(\Lob)\}$ and $O$ and $U$ are arbitrary functions. By Theorem \ref{odds}, for any $\gamma>0$, there exists constants $C_\gamma>0$ and $N_\gamma >0$ such that for any $N\ge N_\gamma$,
			\begin{align*}
				P\left\{\left\|\oddsr(\ \cdot\ ;\hat{\alpha}^r)-\oddsr\right\|_{P,2}\ge C_\gamma\left(\sqrt{\frac{K_r}{N_r}}+K_r^{-\mu_1}\right)\right\}
				\le\gamma\ .
			\end{align*}
			Besides, by \ref{assump-3A}, $\sup_{\lob\in\domr}|u^r_\theta(\lob)-\Phir(\lob)\tr\beta_\theta^r|\le C_2K_r^{-\mu_2}$. So, we consider the set of functions
			\begin{align*}
				\calF_2=\left\{f_{\theta,2}:\left\|O-\oddsr\right\|_{P,2}\le\delta_1',\|u^r_\theta-U\|_\infty\le\delta_2,\theta\in\Theta\right\}
			\end{align*}
			where $\delta_1'=C_\gamma(\sqrt{K_r/N_r}+K_r^{-\mu_1})$ and $\delta_2=C_2K_r^{-\mu_2}$. Then, for any $f_{\theta,2}\in\calF_2$,
			\begin{align*}
				\E\left\{f_{\theta,2}(L,R)\right\}
				&=\E\left[\left\{\mathsf{1}_{R=1_d}\oddsr(\Lob)-\mathsf{1}_{R=r}\right\}\left\{u^r_\theta(\Lob)-U(\Lob)\right\}\right]\\
				&\quads+\E\left[\mathsf{1}_{R=1_d}\left\{O(\Lob)-\oddsr(\Lob)\right\}\left\{u^r_\theta(\Lob)-U(\Lob)\right\}\right]\\
				&\le0+\left\|O-\oddsr\right\|_{P,2}\|u^r_\theta-U\|_{P,2}\\
				&\le\delta_1'\delta_2
				=C_2C_\gamma\left(\frac{K_r^{\frac{1}{2}-\mu_2}}{\sqrt{N_r}}+K_r^{-\mu_1-\mu_2}\right)=o_p(N^{-\frac{1}{2}})\ .
			\end{align*}
			The last line holds due to the fact that $\|\cdot\|_{P,2}\le\|\cdot\|_\infty$ and \ref{assump-3A} and \ref{assump-3E}. 
			Plug in our estimator and define 
			$\hat{f}_{\theta,2}(L,R):=\{\mathsf{1}_{R=1_d}\oddsr(\Lob;\hat{\alpha}^r)-\mathsf{1}_{R=r}\}\{u^r_\theta(\Lob)-\Phir(\Lob)\tr\beta_\theta^r\}$. Then, $\sup_{\theta\in\Theta}|\sqrt{N}S_{\theta,2}^r|\le\sup_{\theta\in\Theta}|\bbG_N(\hat{f}_{\theta,2})|+\sqrt{N}\delta_1'\delta_2$ and 
			\begin{align*}
				P\left(\underset{\theta\in\Theta}{\sup}\left|\sqrt{N}S_{\theta,2}^r\right|>
				\underset{f_{\theta,2}\in\calF_2}{\sup}|\bbG_N(f_{\theta,2})|+\sqrt{N}\delta_1'\delta_2\right)
				\le P\left(\hat{f}_{\theta,2}\notin\calF_2\right)
				\le\gamma\ .
			\end{align*}
			Similarly, we need to show $\E\sup_{f_{\theta,2}\in\calF_2}|\bbG_N(f_{\theta,2})|=o_p(1)$. Define the envelop function $F_2:=(C_0+1)\delta_2$. It's easy to see that $|f_{\theta,2}(L,R)|\le F_2$ for any $f_{\theta,2}\in\calF_2$ when $N$ is large enough. By the maximal inequality with bracketing, 
			\begin{align*}
				\E\underset{f_{\theta,2}\in\calF_2}{\sup}\left|\bbG_N(f_{\theta,2})\right|
				=O_p\left(J_{[ \ ]}\{\|F_2\|_{P,2},\calF_2,L_2(P)\}\right)\ .
			\end{align*}
			To study the entropy integral of $\calF_2$, we first compare it with $\calF_2'=\{f_{\theta,2}:\|O-\oddsr\|_{P,2}\le\delta_1',\|u^r_\theta-U\|_{P,2}\le\delta_2,\theta\in\Theta\}$. It is apparent $\calF_2\subset\calF_2'$. Then, we split function $f_{\theta,2}$ into two parts and consider two sets of functions $\calG_2=\{g_2:\|O-\oddsr\|_{P,2}\le\delta_1'\}$ where $g_2(L,R)=\mathsf{1}_{R=1_d}O(\Lob)-\mathsf{1}_{R=r}$ and $\calH_2=\{h_{\theta,2}:\|h_{\theta,2}\|_{P,2}\le\delta_2,\theta\in\Theta\}$ where $h_{\theta,2}(L)=u^r_\theta(\Lob)-U(\Lob)$. Notice that $\|g_2\|_{P,2}\le C_0+1$ and $\|h_{\theta,2}\|_{P,2}\le\delta_2\le1$ when $N$ is large enough. By Lemma \ref{bracket1},
			\begin{align*}
				n_{[ \ ]}\{4(C_0+2)\epsilon,\calF_2,L_2(P)\}
				\le n_{[ \ ]}\{\epsilon,\calG_2,L_2(P)\}n_{[ \ ]}\{\epsilon,\calH_2,L_2(P)\}\ .
			\end{align*}
			Notice that $\calG_2+(\mathsf{1}_{R=r}-\mathsf{1}_{R=1_d}\oddsr)=\mathsf{1}_{R=1_d}\calG_1$. Since $\mathsf{1}_{R=r}-\mathsf{1}_{R=1_d}\oddsr$ is a fixed function, and $\|\mathsf{1}_{R=1_d}\|_\infty\le1$, by Lemma \ref{bracket4},
			\begin{align*}
				n_{[ \ ]}\{\epsilon,\calG_2,L_2(P)\}
				=n_{[ \ ]}\{\epsilon,\mathsf{1}_{R=1_d}\calG_1,L_2(P)\}
				\le n_{[ \ ]}\{\epsilon,\calG_1,L_2(P)\}\ .
			\end{align*}
			It is obvious that any $\epsilon$-brackets equipped with $\|\cdot\|_\infty$ norm are also $\epsilon$-brackets in $L_2(P)$. With similar arguments in the proof of Lemma \ref{S1}, we have
			\begin{align*}
				n_{[ \ ]}\{\epsilon,\calG_1,L_2(P)\}
				\le n_{[ \ ]}\{\epsilon,\calG_1,L^\infty\}
				\le n_{[ \ ]}\left\{C_\calO\epsilon/\delta_1',\calM^r,L^\infty\right\}\ .
			\end{align*}
			Define a set of functions $\tilde{\calH}_2=\{h_{\theta,2}:\|h_{\theta,2}\|_{P,2}\le C,\theta\in\Theta\}$. Similarly,
			\begin{align*}
				n_{[ \ ]}\{\epsilon,\calH_2,L_2(P)\}
				= n_{[ \ ]}\left\{\epsilon,\delta_2/C\tilde{\calH}_2,L_2(P)\right\}
				=n_{[ \ ]}\left\{C\epsilon/\delta_2,\tilde{\calH}_2,L_2(P)\right\}\ .
			\end{align*}
			Similarly, we split $\tilde{\calH}_2$ into two parts. Define a set of functions $\hat{\calU}^r=\{U:\exists u^r_\theta\in\calU^r \ s.t. \ \|u^r_\theta-U\|_\infty\le C,\}$ where $\calU^r=\{u^r_\theta:\theta\in\Theta\}$. By Lemma \ref{bracket2},
			\begin{align*}
				n_{[ \ ]}\{2\epsilon,\tilde{\calH}_2,L_2(P)\}
				\le n_{[ \ ]}\{\epsilon,\calU^r,L_2(P)\}n_{[ \ ]}\{\epsilon,\hat{\calU}^r,L_2(P)\}\ .
			\end{align*}
			Also define a set of functions $\E\calH^r:=\{g^r(\lob):=\E\{f(L)\mid\Lob=\lob,R=r\},f\in\calH\}$. Although the set $\calU^r$ is unknown, we should not consider much more rough functions than those in $\E\calH^r$. Therefore, there exists a constant $C_{\hat{\calU}^r}$ such that $\hat{\calU}^r\subset\E\calH^r$. Thus, by Lemma \ref{bracket3},
			\begin{align*}
				n_{[ \ ]}\{\epsilon,\hat{\calU}^r,L_2(P)\}
				\le n_{[ \ ]}\{\epsilon,\E\calH^r,L_2(P)\}
				\le n_{[ \ ]}\{\epsilon,\calH,L_2(P)\}\ .
			\end{align*}
			By \ref{assump-3B}, \ref{assump-3C}, and the above inequalities,
			\begin{align*}
				&\quads J_{[ \ ]}\{\|F_2\|_{P,2},\calF_2,L_2(P)\}\\
				&\le\int_0^{(C_0+1)\delta_2}\sqrt{\log n_{[ \ ]}\left\{\frac{C_\calO\epsilon}{4(C_0+2)\delta_1'},\calM^r,L_2(P)\right\}}d\epsilon\\
				&\quads+\sqrt{2}\int_0^{(C_0+1)\delta_2}\sqrt{\log n_{[ \ ]}\left\{\frac{C_{\hat{\calU}^r}\epsilon}{8(C_0+2)\delta_2},\calH,L_2(P)\right\}}d\epsilon\\
				&\le\sqrt{C_\calM}\{4(C_0+2)\delta_1'/C_\calO\}^{\frac{1}{2d_\calM}}\{(C_0+1)\delta_2\}^{1-\frac{1}{2d_\calM}}\\
				&\quads+\sqrt{2C_\calH}\{8(C_0+2)/C_{\hat{\calU}^r}\}^{\frac{1}{2d_\calH}}(C_0+1)^{1-\frac{1}{2d_\calH}}\delta_2\\
				&\to0
			\end{align*}
			since $d_\calM,d_\calH>1/2$ and $\delta_1',\delta_2\to0$ as $N\to\infty$. So, $\E\sup_{f_{\theta,2}\in\calF_2}|\bbG_N(f_{\theta,2})|=O_p(o_p(1))=o_p(1)$ and $\sup_{\theta\in\Theta}|\sqrt{N}S_{\theta,2}^r|=o_p(1)$.
		\end{proof}

		\begin{lemma}\label{S3}
			Under Assumptions \ref{assump1}--\ref{assump3}, for any missing pattern $r$, 
			$$
			\sup_{\theta\in\Theta}|\sqrt{N}S_{\theta,3}^r|=o_p(1).$$
		\end{lemma}
		
		\begin{proof}
			Notice that $S_{\theta,3}^r$ is related to the balancing error:
			\begin{align*}
				\underset{\theta\in\Theta}{\sup}\left|\sqrt{N}S_{\theta,3}^r\right|
				&=\underset{\theta\in\Theta}{\sup}\left|\frac{1}{N}\sum_{i=1}^N\left\{\mathsf{1}_{R_i=1_d}\oddsr(\Lobi;\hat{\alpha}^r)-\mathsf{1}_{R_i=r}\right\}\Phir(\Lob_i)\tr\beta_\theta^r\right|\\
				&\le\lambda\left\{\gamma\sqrt{K_r}+2(1-\gamma)\sqrt{\PEN_2({\Phir}\tr\hat{\alpha}^r)}\right\}\sqrt{\PEN_2({\Phir}\tr\beta_\theta^r)}
			\end{align*} 
			where $\Phir(\lob)\tr\hat{\alpha}^r=\log\oddsr(\lob;{\hat{\alpha}^r})$ denotes the log transformation of the propensity odds model. Due to the similar reason, the roughness of the approximation functions are bounded. Besides, by \ref{assump-2D}, $\lambda=o(1/\sqrt{K_rN_r})$. Thus, $\sup_{\theta\in\Theta}|\sqrt{N}S_{\theta,3}^r|=o_p(1)$.
		\end{proof}

		\begin{lemma} \label{uniformbound}
			Suppose that Assumptions \ref{assump1}--\ref{assump4} hold. Then,
			\begin{align*}
				\underset{\theta\in\Theta}{\sup}\left|\hat{\bbP}_N\psi_\theta-\E\{\psi_\theta(L)\}\right|=o_p(1)\ .
			\end{align*}
		\end{lemma}
		
		\begin{proof}
			By Lemmas \ref{S1}, \ref{S2}, and \ref{S3}, we only need to show $\sup_{\theta\in\Theta}|S_{\theta,4}^r|=o_p(1)$ where
			\begin{align*}
				S_{\theta,4}^r
				&=\frac{1}{N}\sum_{i=1}^N\mathsf{1}_{R_i=1_d}\oddsr(\Lob_i)\left\{\psi_\theta(L_i)-u^r_\theta(\Lob_i)\right\}\\
				&\quads+\frac{1}{N}\sum_{i=1}^N\mathsf{1}_{R_i=r}u^r_\theta(\Lob_i)-\E\{\mathsf{1}_{R=r}\psi_\theta(L)\}\ .
			\end{align*}
			Study the following decomposition. Let $\calF_a=\{f_{\theta,a}:\theta\in\Theta\}$ where $f_{\theta,a}(L,R)=\mathsf{1}_{R=r}u^r_\theta(\Lob)$. It's easy to see that for any $\epsilon>0$,
			\begin{align*}
				n_{[ \ ]}\{\epsilon,\calF_a,L_2(P)\}
				\le n_{[ \ ]}\{\epsilon,\calU^r,L_2(P)\}
				\le n_{[ \ ]}\{\epsilon,\calH,L_2(P)\}<\infty\ .
			\end{align*}
			For any measurable function $f$, $\|f(L)\|_{P,2}^2=\E\{f(L)^2\}\ge\{\E|f(L)|\}^2=\|f\|_{P,1}^2$. Thus, 
			\begin{align*}
				n_{[ \ ]}\{\epsilon,\calF_a,L_1(P)\}\le n_{[ \ ]}\{\epsilon,\calF_a,L_2(P)\}\ .
			\end{align*}
			By Theorem 19.4 in \cite{van2000asymptotic}, $\calF_a$ is Glivenko-Cantelli. Thus,
			\begin{align*}
				\underset{\theta\in\Theta}{\sup}\left|\bbP_Nf_{\theta,a}-Pf_{\theta,a}\right|\xrightarrow{a.s.}0\ .
			\end{align*}
			Also let $\calF_b=\{f_{\theta,b}:\theta\in\Theta\}$ where $f_{\theta,b}(L,R)=\mathsf{1}_{R=1_d}\oddsr(\Lob)\psi_\theta(L)$ and $\calF_c=\{f_{\theta,c}:\theta\in\Theta\}$ where $f_{\theta,b}(L,R)=\mathsf{1}_{R=1_d}\oddsr(\Lob)u^r_\theta(\Lob)$. Similarly,
			\begin{align*}
				\underset{\theta\in\Theta}{\sup}\left|\bbP_Nf_{\theta,b}-Pf_{\theta,b}\right|\xrightarrow{a.s.}0
				\textrm{ and }
				\underset{\theta\in\Theta}{\sup}\left|\bbP_Nf_{\theta,c}-Pf_{\theta,c}\right|\xrightarrow{a.s.}0\ .
			\end{align*}
			Notice that $\E\{f_{\theta,b}(L,R)\}=\E\{f_{\theta,c}(L,R)\}$. Besides, the convergence almost surely implies the convergence in probability. Thus, $\sup_{\theta\in\Theta}|S_{\theta,4}^r|=o_p(1)$. Then,
			\begin{align*}
				\underset{\theta\in\Theta}{\sup}\left|\hat{\bbP}_N\psi_\theta-\E\{\psi_\theta(L)\}\right|=o_p(1)\ .
			\end{align*}
		\end{proof}

		\begin{lemma} \label{S5}
			Under Assumptions \ref{assump1}--\ref{assump4}, we have
			\begin{align*}
				\sqrt{N}\left|S_{\hat{\theta}_N,5}^r-S_{\theta_0,5}^r\right|=o_p(1)\ .
			\end{align*}
		\end{lemma}
		
		\begin{proof}
			Consider the following empirical process.
			\begin{align*}
				\bbG_N(f_{\theta,5})=\sqrt{N}\left[\frac{1}{N}\sum_{i=1}^Nf_{\theta,5}(L_i,R_i)-\E\left\{f_{\theta,5}(L,R)\right\}\right]
			\end{align*}
			where $f_{\theta,5}(L,R)=\{\mathsf{1}_{R=1_d}\oddsr(\Lob)-\mathsf{1}_{R=r}\}\{\psi_\theta(L)-\psi_{\theta_0}(L)\}$. Pick any decreasing sequence $\{\delta_m\}\to0$. Since $\|\hat{\theta}_N-\theta_0\|_2=o_p(1)$, for any $\gamma>0$ and each $\delta_m$, there exists a constant $N_{\delta_m,\gamma}>0$ such that for any $N\ge N_{\delta_m,\gamma}$,
			\begin{align}\label{ineq-theta}
				P\left(\|\hat{\theta}_N-\theta_0\|_2\ge\delta_m\right)\le\gamma
			\end{align}
			Consider the set of functions $\calF_5=\{f_{\theta,5}:\|\theta-\theta_0\|_2\le\delta_m\}$. It is easy to check that $\E\{f_{\theta,5}(L,R)\}=0$. Plug in our estimator and define $\hat{f}_{\theta,5}(L,R):=\{\mathsf{1}_{R=1_d}\oddsr(\Lob)-\mathsf{1}_{R=r}\}\{\psi_{\hat{\theta}_N}(L)-\psi_{\theta_0}(L)\}$. Notice that $\sqrt{N}(S_{\hat{\theta}_N,5}^r-S_{\theta_0,5}^r)=\bbG_N(\hat{f}_{\theta,5})$. Thus, 
			\begin{align*}
				P\left(\sqrt{N}\left|S_{\hat{\theta}_N,5}^r-S_{\theta_0,5}^r\right|>\underset{f_{\theta,5}\in\calF_5}{\sup}|\bbG_N(f_{\theta,5})|\right)
				\le P\left(\hat{f}_{\theta,5}\notin\calF_5\right)
				\le\gamma\ .
			\end{align*}
			Similarly, we only need to show $\E\sup_{f_{\theta,5}\in\calF_5}|\bbG_N(f_{\theta,5})|=o_p(1)$. Define the envelop function $F_5(L):=(C_0+1)f_{\delta_m}(L)$ where $f_\delta$ is the envelop function in \ref{assump-4B}. So, $|f_{\theta,5}(L,R)|\le F_5(L)$ for any $f_{\theta,5}\in\calF_5$. Besides, $\|F_5\|_{P,2}\le (C_0+1)\|f_{\delta_m}\|_{P,2}$. Due to the maximal inequality,
			\begin{align*}
				\E\underset{f_{\theta,5}\in\calF_5}{\sup}\left|\bbG_N(f_{\theta,5})\right|
				=O_p\left(J_{[ \ ]}\{\|F_5\|_{P,2},\calF_5,L_2(P)\}\right)\ .
			\end{align*}
			Define a set of functions $\calG_5=\{g_{\theta,5}:\|\theta-\theta_0\|_2\le\delta_m\}$ where $g_{\theta,5}(L)=\psi_\theta(L)-\psi_{\theta_0}(L)$. Since $\|\mathsf{1}_{R=1_d}\oddsr-\mathsf{1}_{R=r}\|_\infty\le(C_0+1)$, by Lemma \ref{bracket4},
			\begin{align*}
				n_{[ \ ]}\{(C_0+1)\epsilon,\calF_5,L_2(P)\}
				\le n_{[ \ ]}\{\epsilon,\calG_5,L_2(P)\}\ .
			\end{align*}
			Define a set of functions $\tilde{\calG}_5=\{\psi_\theta:\|\theta-\theta_0\|_2\le\delta_m\}$. Since $\psi_{\theta_0}$ is a fixed function, $n_{[ \ ]}\{\epsilon,\calG_5,L_2(P)\}
			=n_{[ \ ]}\{\epsilon,\tilde{\calG}_5,L_2(P)\}$. Since $\delta_m\to0$ as $N\to\infty$, we can take $\delta_m$ small enough such that the set $\{\theta:\|\theta-\theta_0\|_2\le\delta_m\}\subset\Theta$. So, $\tilde{\calG}_5\subset\calH$, and $n_{[ \ ]}\{\epsilon,\tilde{\calG}_5,L_2(P)\}
			\le n_{[ \ ]}\{\epsilon,\calH,L_2(P)\}$. Then,
			\begin{align*}
				J_{[ \ ]}\{\|F_5\|_{P,2},\calF_5,L_2(P)\}
				&\le\int_0^{(C_0+1)\|f_{\delta_m}\|_{P,2}}\sqrt{\log n_{[ \ ]}\left\{\frac{\epsilon}{C_0+1},\calH,L_2(P)\right\}}d\epsilon\\
				&\le\sqrt{C_\calH}\int_0^{(C_0+1)\|f_{\delta_m}\|_{P,2}}\{(C_0+1)/\epsilon\}^{\frac{1}{2d_\calH}}d\epsilon\\
				&\le\sqrt{C_\calH}(C_0+1)\|f_{\delta_m}\|_{P,2}^{1-\frac{1}{2d_\calH}}\\
				&\to0
			\end{align*}
			since $d_\calH>1/2$ and $\|f_{\delta_m}\|_{P,2}\to0$ as $N\to\infty$. Thus, $\E\sup_{f_{\theta,5}\in\calF_5}|\bbG_N(f_{\theta,5})|=o_p(1)$ and $\sqrt{N}|S_{\hat{\theta}_N,5}^r-S_{\theta_0,5}^r|=o_p(1)$.
		\end{proof}

		\begin{lemma} \label{S6}
			Under Assumptions \ref{assump1}--\ref{assump4}, we have
			\begin{align*}
				\sqrt{N}\left|S_{\hat{\theta}_N,6}^r-S_{\theta_0,6}^r\right|=o_p(1)\ .
			\end{align*}
		\end{lemma}
		
		\begin{proof}
			Consider the following empirical process.
			\begin{align*}
				\bbG_N(f_{\theta,6})=\sqrt{N}\left[\frac{1}{N}\sum_{i=1}^Nf_{\theta,6}(L_i,R_i)-\E\left\{f_{\theta,6}(L,R)\right\}\right]
			\end{align*}
			where $f_{\theta,6}(L,R)=\{\mathsf{1}_{R=1_d}\oddsr(\Lob)-\mathsf{1}_{R=r}\}\{u^r_\theta(\Lob)-u^r_{\theta_0}(\Lob)\}$. Similarly, inequality \eqref{ineq-theta} holds and $\E\{f_{\theta,6}(L,R)\}=0$. Consider the set of functions $\calF_6=\{f_{\theta,6}:\|\theta-\theta_0\|_2\le\delta_m\}$. To show $\sqrt{N}|S_{\hat{\theta}_N,6}^r-S_{\theta_0,6}^r|=o_p(1)$, we need to show $\E\sup_{f_{\theta,6}\in\calF_6}|\bbG_N(f_{\theta,6})|=o_p(1)$. Define the envelop function $F_6(L):=(C_0+1)\E\{f_{\delta_m}(L)\mid\Lob\}$. It's easy to see that $|f_{\theta,6}(L,R)|\le F_6(L)$ for any $f_{\theta,6}\in\calF_6$ and $\|F_6\|_{P,2}\le (C_0+1)\|f_{\delta_m}\|_{P,2}$. Apply the maximal inequality,
			\begin{align*}
				\E\underset{f_{\theta,6}\in\calF_6}{\sup}\left|\bbG_N(f_{\theta,6})\right|
				=O_p\left(J_{[ \ ]}\{\|F_6\|_{P,2},\calF_6,L_2(P)\}\right)\ .
			\end{align*}
			Define a set of functions $\calG_6=\{g_{\theta,6}:\|\theta-\theta_0\|_2\le\delta\}$ where $g_{\theta,6}(L)=u^r_\theta(\Lob)-u^r_{\theta_0}(\Lob)$. Since $\|\mathsf{1}_{R=1_d}\oddsr-\mathsf{1}_{R=r}\|_\infty\le(C_0+1)$, by Lemma \ref{bracket4},
			\begin{align*}
				n_{[ \ ]}\{(C_0+1)\epsilon,\calF_6,L_2(P)\}
				\le n_{[ \ ]}\{\epsilon,\calG_6,L_2(P)\}\ .
			\end{align*}
			Define a set of functions $\tilde{\calG}_6=\{u^r_\theta:\|\theta-\theta_0\|_2\le\delta\}$. Similarly, since $u^r_{\theta_0}$ is a fixed function, $n_{[ \ ]}\{\epsilon,\calG_6,L_2(P)\}=n_{[ \ ]}\{\epsilon,\tilde{\calG}_6,L_2(P)\}$. Take $\delta_m$ small enough such that the set $\{\theta:\|\theta-\theta_0\|_2\le\delta_m\}\subset\Theta$. Then, $\tilde{\calG}_6\subset\calU^r$, and by Lemma \ref{bracket3},
			\begin{align*}
				n_{[ \ ]}\{\epsilon,\tilde{\calG}_6,L_2(P)\}
				\le n_{[ \ ]}\{\epsilon,\calU^r,L_2(P)\}
				\le n_{[ \ ]}\{\epsilon,\calH,L_2(P)\}\ .
			\end{align*}
			Therefore,
			\begin{align*}
				J_{[ \ ]}\{\|F_6\|_{P,2},\calF_6,L_2(P)\}
				&\le\int_0^{(C_0+1)\|f_{\delta_m}\|_{P,2}}\sqrt{\log n_{[ \ ]}\left(\epsilon/(C_0+1),\calH,L_2(P)\right)}d\epsilon\\
				&\le\sqrt{C_\calH}(C_0+1)\|f_{\delta_m}\|_{P,2}^{1-\frac{1}{2d_\calH}}\\
				&\to0
			\end{align*}
			since $d_\calH>1/2$ and $\|f_{\delta_m}\|_{P,2}\to0$ as $N\to\infty$. Thus, $\E\sup_{f_{\theta,6}\in\calF_6}|\bbG_N(f_{\theta,6})|=O_p(o_p(1))=o_p(1)$ and $\sqrt{N}|S_{\hat{\theta}_N,6}^r-S_{\theta_0,6}^r|=o_p(1)$.
		\end{proof}

		\begin{lemma} \label{bracket1}
			Consider the set of functions $\calF=\{f:=gh,g\in\calG,h\in\calH\}$. Assume that $\|g\|_\infty\le c_g$ for all $g\in\calG$ and $\|h\|_{P,2}\le c_h$ for all $h\in\calH$.
			Then, for any $\epsilon\le\min\{c_g,c_h\}$,
			\begin{align*}
				n_{[ \ ]}\{4(c_g+c_h)\epsilon,\calF,L_2(P)\}
				\le n_{[ \ ]}\{\epsilon,\calG,L^\infty\}n_{[ \ ]}\{\epsilon,\calH,L_2(P)\}\ .
			\end{align*}
		\end{lemma}
		
		\begin{proof}
			Suppose $\{u_i,v_i\}_{i=1}^n$ are the $\epsilon$-brackets that can cover $\calG$ and $\{U_j,V_j\}_{j=1}^m$ are the $\epsilon$-brackets that can cover $\calH$. Define the bracket $[\mathsf{U}_k,\mathsf{V}_k]$ for $k=(i-1)m+j$ where $i=1,\cdots,n,j=1,\cdots,m$:
			\begin{align*}
				\mathsf{U}_k(x)=\min\{u_i(x)U_j(x),u_i(x)V_j(x),v_i(x)U_j(x),v_i(x)V_j(x)\}\ ,\\
				\mathsf{V}_k(x)=\max\{u_i(x)U_j(x),u_i(x)V_j(x),v_i(x)U_j(x),v_i(x)V_j(x)\}\ .
			\end{align*}
			For any function $f\in\calF$, there exists functions $g\in\calG$ and $h\in\calH$ such that $f=gh$. Besides, we can find two pairs of functions $(u_{i_0},v_{i_0})$ and $(U_{j_0},V_{j_0})$ such that $u_{i_0}(x)\le g(x)\le v_{i_0}(x)$, $U_{j_0}(x)\le h(x)\le V{j_0}(x)$, $\|u_{i_0}-v_{i_0}\|_\infty\le\epsilon$, and $\|U_{j_0}-V_{j_0}\|_{P,2}\le\epsilon$. Then, $\mathsf{U}_{k_0}(x)\le f(x)\le \mathsf{V}_{k_0}(x)$ where $k_0=(i_0-1)m+j_0$. Then, we look at the size of the new brackets. By simple algebra,
			\begin{align*}
				\|\mathsf{U}_k-\mathsf{V}_k\|_{P,2}
				&\le\left\|\left(|u_i|+|v_i|\right)|U_j-V_j|+\left(|U_j|+|V_j|\right)|u_i-v_i|\right\|_{P,2}\\
				&\le\|u_i\|_\infty\|U_j-V_j\|_{P,2}+\|v_i\|_\infty\|U_j-V_j\|_{P,2}\\
				&+\|u_i-v_i\|_\infty\|U_j\|_{P,2}+\|u_i-v_i\|_\infty\|V_j\|_{P,2}\\
				&\le2\epsilon(c_g+\epsilon)+2(c_h+\epsilon)\epsilon=2(c_g+c_h+2\epsilon)\epsilon\ .
			\end{align*}
			Furthermore, for any $\epsilon\le\min\{c_g,c_h\}$, we have $2(c_g+c_h+2\epsilon)\epsilon\le 4(c_g+c_h)\epsilon$.
			Therefore,
			\begin{align*}
				n_{[ \ ]}\{4(c_g+c_h)\epsilon,\calF,L_2(P)\}
				\le n_{[ \ ]}\{\epsilon,\calG,L^\infty\}n_{[ \ ]}\{\epsilon,\calH,L_2(P)\}\ .
			\end{align*}
		\end{proof}

		\begin{lemma} \label{bracket2}
			Consider the set of functions $\calF=\calH+\calG=\{f:=g+h,g\in\calG,h\in\calH\}$. Assume that $\|g\|_{P,2}\le c_g$ for all $g\in\calG$ and $\|h\|_{P,2}\le c_h$ for all $h\in\calH$.
			Then,
			\begin{align*}
				n_{[ \ ]} \{2\epsilon,\mathcal{F},L_2(P)\}
				\le n_{[ \ ]}\{\epsilon,\calG,L_2(P)\}n_{[ \ ]}\{\epsilon,\calH,L_2(P)\}\ .
			\end{align*}
		\end{lemma}
		
		\begin{proof}
			Suppose $\{u_i,v_i\}_{i=1}^n$ are the $\epsilon$-brackets that can cover $\calG$ and $\{U_j,V_j\}_{j=1}^m$ are the $\epsilon$-brackets that can cover $\calH$. Define the bracket $[\mathsf{U}_k,\mathsf{V}_k]$ for $k=(i-1)m+j$ and $i=1,\cdots,n,j=1,\cdots,m$:
			\begin{align*}
				\mathsf{U}_k(x)=u_i(x)+U_j(x)\ ,\\
				\mathsf{V}_k(x)=v_i(x)+V_j(x)\ .
			\end{align*}
			For any function $f\in\calF$, there exists functions $g\in\calG$ and $h\in\calH$ such that $f=g+h$. Besides, we can find two pairs of functions $(u_{i_0},v_{i_0})$ and $(U_{j_0},V_{j_0})$ such that $u_{i_0}(x)\le g(x)\le v_{i_0}(x)$, $U_{j_0}(x)\le h(x)\le V{j_0}(x)$, $\|u_{i_0}-v_{i_0}\|_{P,2}\le\epsilon$, and $\|U_{j_0}-V_{j_0}\|_{P,2}\le\epsilon$. Then, $\mathsf{U}_{k_0}(x)\le f(x)\le \mathsf{V}_{k_0}(x)$ where $k_0=(i_0-1)m+j_0$ and
			\begin{align*}
				\|\mathsf{U}_k-\mathsf{V}_k\|_{P,2} 
				\le\|u_i-v_i\|_{P,2}+\|U_j-V_j\|_{P,2}
				\le 2\epsilon\ .
			\end{align*}
			Therefore,
			\begin{align*}
				n_{[ \ ]}\{2\epsilon,\mathcal{F},L_2(P)\}
				\le n_{[ \ ]}\{\epsilon,\calG,L_2(P)\}n_{[ \ ]}\{\epsilon,\calH,L_2(P)\}\ .
			\end{align*}
		\end{proof}

		\begin{lemma} \label{bracket3}
			Let $\calH$, $\calU^r$ and $\E\calH^r$ be the sets of functions as we defined before. Then, 
			\begin{align*}
				n_{[ \ ]}\{\epsilon,\calU^r,L_2(P)\}
				&\le n_{[ \ ]}\{\epsilon,\calH,L_2(P)\}\ ,\\
				n_{[ \ ]}\{\epsilon,\E\calH^r,L_2(P)\}
				&\le n_{[ \ ]}\{\epsilon,\calH,L_2(P)\}\ .
			\end{align*}
		\end{lemma}
		
		\begin{proof}
			Suppose $\{u_i,v_i\}_{i=1}^n$ are the $\epsilon$-brackets that can cover $\calH$. Define $U_i(\lob)=\E\{u_i(L)\mid\Lob=\lob,R=r\}$ and $V_i(\lob)=\E\{v_i(L)\mid\Lob=\lob,R=r\}$. Then, for any $u^r\in\calU^r$, there exists $\psi_\theta\in\calH$ such that $u^r(\lob)=\E\{\psi_\theta(L)\mid\Lob=\lob,R=r\}$ with a pair of functions $(u_0,v_0)$ satisfying $u_0(l)\le\psi_\theta(l)\le v_0(l)$ and $\|u_0(L)-v_0(L)\|_{P,2}\le\epsilon$. Then, $U_0(\lob)\le u^r(\lob)\le V_0(\lob)$ and 
			\begin{align*}
				\|U_0(\Lob)-V_0(\Lob)\|_{P,2}
				&=\E[\E\{u_0(L)-v_0(L)\mid\Lob=\lob,R=r\}^2]\\
				&\le\E\E[\{u_0(L)-v_0(L)\}^2\mid\Lob=\lob,R=r]\\
				&=\E\{u_0(L)-v_0(L)\}^2
				=\|u_0-v_0\|_{P,2}
				\le\epsilon\ .
			\end{align*}
			So, $\{U_i,V_i\}_{i=1}^n$ are the $\epsilon$-brackets that can cover $\calU^r$ and
			\begin{align*}
				n_{[ \ ]}\{\epsilon,\calU^r,L_2(P)\} 
				\le n_{[ \ ]}\{\epsilon,\calH,L_2(P)\}\ .
			\end{align*}
			For any $g^r\in\E\calH^r$, there exists $f\in\calH$ such that $g^r(\lob)=\E\{f(L)\mid\Lob=\lob,R=r\}$. Similarly,
			\begin{align*}
				n_{[ \ ]}\{\epsilon,\E\calH^r,L_2(P)\} 
				\le n_{[ \ ]}\{\epsilon,\calH,L_2(P)\}\ .
			\end{align*}
		\end{proof}

		\begin{lemma} \label{bracket4}
			Let $h$ be a fixed bounded function. Assume $\|h\|_\infty\le c_h$. We consider two function classes $\calF=\{f:f(x):=g(x)h(x),g\in\calG\}$ and $\calG=\{g:\|g\|_{P,2}\le c\}$ for a fixed constant $c$. Then, 
			\begin{align*}
				n_{[ \ ]}\{ c_h\epsilon,\calF,L_2(P)\} 
				\le n_{[ \ ]}\{\epsilon,\calG,L_2(P)\}\ .
			\end{align*}
		\end{lemma}
		
		\begin{proof}
			Suppose $\{u_i,v_i\}_{i=1}^n$ are the $\epsilon$-brackets that can cover $\calG$. That is, for any $g\in\calG$, we can find a pair of functions $(u_0,v_0)$ such that $u_0(x)\le g(x)\le v_0(x)$ and $\|u_0-v_0\|_{P,2}\le\epsilon$. Then, for any $x$, either $u_0(x)h(x)\le g(x)h(x)\le v_0(x)h(x)$ or $u_0(x)h(x)\ge g(x)h(x)\ge v_0(x)h(x)$ holds. Define $U_i(x)=\min\{u_i(x)h(x),v_i(x)h(x)\}$ and $V_i(x)=\max\{u_i(x)h(x),v_i(x)h(x)\}$. For any $f\in\calF$, there exists $g\in\calG$ such that $f=gh$ and a pair of functions $(U_0,V_0)$ such $U_0(x)\le f(x)\le V_0(x)$ and 
			\begin{align*}
				\|U_i-V_i\|_{P,2}
				=\|(u_i-v_i)h\|_{P,2}
				\le\|h\|_\infty\|(u_i-v_i)\|_{P,2}
				\le c_h\epsilon\ .
			\end{align*}
			So, $\{U_i,V_i\}_{i=1}^n$ are the $c_h\epsilon$-brackets that can cover $\calF$ and
			\begin{align*}
				n_{[ \ ]}\{c_h\epsilon,\calF,L_2(P)\} 
				\le n_{[ \ ]}\{\epsilon,\calG,L_2(P)\}\ .
			\end{align*}
		\end{proof}
	\end{appendices}

	\bibliographystyle{chicago}
	\bibliography{refer.bib}

\end{document}